%% file: main-split.tex
\pdfoutput=1
\newif\ifdraft
\draftfalse
\documentclass[letterpaper,11pt]{amsart}
\usepackage{amsmath}
\usepackage{amsthm}
\usepackage{amssymb}
\usepackage{amsfonts}

\usepackage{amsaddr}
\makeatletter
\renewcommand{\email}[2][]{%
  \ifx\emails\@empty\relax\else{\g@addto@macro\emails{,\space}}\fi%
  \@ifnotempty{#1}{\g@addto@macro\emails{\textrm{(#1)}\space}}%
  \g@addto@macro\emails{#2}%
}
\makeatother

\usepackage{fullpage}
\usepackage{algorithm}
\usepackage{graphicx,color}
\usepackage[dvipsnames]{xcolor}
\usepackage{fancybox}
\usepackage{enumerate}
\usepackage{tikz}
\usepackage{subcaption}
\usepackage{cite}
\usepackage[utf8]{inputenc}
\usepackage{bm}
\renewcommand{\mathbf}[1]{\bm{#1}}

\usepackage{cool}
\Style{DSymb={\mathrm d},DShorten=false,IntegrateDifferentialDSymb=\mathrm{d}}

\newcommand{\knuth}[1]{\ensuremath{\inb{#1}}}
\ifdraft
\usepackage{todonotes}
\usepackage{scrtime}
\usepackage{ifthen}
\usepackage{prelim2e}
\newcommand{\datetime}{\the\year-\ifthenelse{\the\month < 10}{0}{}\the\month-\ifthenelse{\the\day < 10}{0}{}\the\day{} \thistime}

\fi

\ifdraft
\usepackage[pagebackref]{hyperref}
\else
  \usepackage[pagebackref]{hyperref}
\fi
\hypersetup{
  colorlinks=true,
  linkcolor=green!30!black,
  linktoc=all,
  citecolor=blue,
  bookmarksnumbered=true
}
\usepackage{cleveref}

\newtheorem{theorem}{Theorem}[section]

\newtheorem{corollary}[theorem]{Corollary}
\newtheorem{lemma}[theorem]{Lemma}

\theoremstyle{definition}
\newtheorem{definition}[theorem]{Definition}

\newcommand{\abs}[1]{\left\vert#1\right\vert}
\newcommand{\set}[1]{\left\{#1\right\}}
\newcommand{\eps}{\varepsilon}

\newcommand{\defeq}{:=}

\newcommand{\Aout}{A_{\mathrm{out}}}
\newcommand{\F}{\mathcal{F}}

\renewcommand{\E}{\mathbb{E}}

\newcommand{\polylog}[1]{\ensuremath{\mathop{\mathrm{polylog}}\inp{#1}}}

\newcommand{\CommentS}[1]{}

\newcommand{\supp}{\mathrm{Supp}}
\newcommand{\G}{\mathcal{G}}

\newcommand{\R}{\mathbb{R}}

\newcommand{\M}{\mathcal{M}}

\usepackage{mathtools}

\DeclareMathOperator*{\argmax}{arg\,max}

\newcommand{\ed}{\ensuremath{(\varepsilon, \delta)}}

\newcommand{\Median}{\mathrm{Median}}
\newcommand{\Qp}{Q^{(p)}}

\newcommand{\EM}{\mathcal{E}}
\newcommand{\Data}{D} 
\newcommand{\dis}[2]{\mathrm{D}\inp{ #1 \Vert #2}}

\newcommand{\disd}[2]{\mathrm{D}_{\delta}\inp{ #1 \Vert #2}}
\newcommand{\T}{\tau}
\newcommand{\Tt}{\widetilde{\tau}}
\newcommand{\qt}{\widetilde{q}}
\newcommand{\xt}{\widetilde{x}}
\newcommand{\Lap}{\mathtt{Lap}}
\newcommand{\Tstar}{\T_*}

\newcommand{\p}{p}
\newcommand{\pt}{\widetilde{p}}
\newcommand{\pti}{\widetilde{p_i}}
\newcommand{\ptk}{\widetilde{p_k}}
\newcommand{\pstar}{\p_*}

\newcommand{\Pnd}{\Phi^{(N,\Delta)}}

\crefname{theorem}{Theorem}{Theorems}
\crefname{observation}{Observation}{Observations}
\crefname{claim}{Claim}{Claims}
\crefname{condition}{Condition}{Conditions}
\crefname{algorithm}{Algorithm}{Algorithms}
\crefname{property}{Property}{Properties}
\crefname{example}{Example}{Examples}
\crefname{fact}{Fact}{Facts}
\crefname{lemma}{Lemma}{Lemmas}
\crefname{corollary}{Corollary}{Corollaries}
\crefname{definition}{Definition}{Definitions}
\crefname{remark}{Remark}{Remarks}
\crefname{proposition}{Proposition}{Propositions}
\crefname{equation}{eq.}{eqs.}

\newcommand{\liuexp}[1]{}

\begin{document}

\title{Private Selection from Private Candidates}
\author{Jingcheng Liu \and Kunal Talwar}

\date{}
\maketitle
\begin{abstract}
	Differentially Private algorithms often need to select the best amongst many candidate options. Classical works on this {\em selection} problem require that the candidates' goodness, measured as a real-valued score function, does not change by much when one person's data changes. In many applications such as hyperparameter optimization, this stability assumption is much too strong. In this work, we consider the selection problem under a much weaker stability assumption on the candidates, namely that the score functions are differentially private. Under this assumption, we present algorithms that are near-optimal along the three relevant dimensions: privacy, utility and computational efficiency.

	Our result can be seen as a generalization of the exponential mechanism and its existing generalizations. We also develop an online version of our algorithm, that can be seen as a generalization of the sparse vector technique to this weaker stability assumption. We show how our results imply  better algorithms for hyperparameter selection in differentially private machine learning, as well as for adaptive data analysis.

\end{abstract}
\let\bakthefootnote\thefootnote
\let\thefootnote\relax
\footnotetext{Jingcheng Liu, Computer Science Division, UC Berkeley. Email: \texttt{liuexp@berkeley.edu}.}
\footnotetext{Kunal Talwar, Google Brain. Email: \texttt{kunal@google.com}.}
\footnotetext{Some of this work was done while JL was an intern at Google Brain.}
\let\thefootnote\bakthefootnote
\newpage
\setcounter{page}{1}

\input{intro}
\input{prelims}

\input{priv_select}
\input{general_sparse_vector}

\input{sparse_vector}
\input{applications}
\section{Conclusions}
We have presented new differentially private algorithms for selecting the best amongst several differentially private algorithms. Our algorithm is near-optimal in terms of privacy overhead, computational cost and utility loss. We have shown how it applies to hyperparameter search and adaptive data analysis. We leave open the question of improving the constants in the run time of our threshold finding algorithm.

While random search is a surprisingly effective way to do hyperparameter optimization in machine learning~\cite{hyperband}, there are more complex adaptive algorithms that often do better. Our work says that random search- or grid search-based hyperparameter tuning can be made differentially private essentially for free. It is natural to ask if we can make the various adaptive algorithms differentially private.
\bibliographystyle{abbrv}
\bibliography{refs}

\appendix
\input{deferred}
\input{naive_algs}

\input{priv_amp}

\input{lower_bounds}

\input{app_dp_properties}

\end{document}

%% file: intro.tex
\section{Introduction}

Differential Privacy~\cite{DMNS} is the standard notion of privacy for statistical databases. It imposes a probabilistic constraint on the behavior of the algorithm on datasets that differ in one person's input. Formally,
\begin{definition}[Differential Privacy]
	Let $\mathcal{M} : \mathcal{D}^n \rightarrow \mathcal{R}$ be a randomized algorithm mapping datasets to some range $\mathcal{R}$. We say that $\mathcal{M}$ is $\ed$-differentially private if for all pairs of adjacent datasets $D, D' \in \mathcal{D}^n$, and for all measurable subsets $S \subseteq \mathcal{R}$,
	\begin{align*}
		\Pr[\mathcal{M}(D) \in S] \leq \exp(\eps) \cdot \Pr[\mathcal{M}(D') \in S] + \delta.
	\end{align*}
	Here, two datasets are adjacent if they differ in one person's input. When $\delta=0$, we will sometimes say that $\mathcal{M}$ is $\eps$-differentially private.
\end{definition}
Differential privacy (DP) satisfies nice post-processing and composition properties, allowing for complex differentially private algorithms to be built out of simpler building blocks. In the last decade or so, differentially private algorithms have been designed and analyzed for numerous statistical and machine learning tasks, in most cases by carefully putting together these building blocks. This approach to the design and analysis of differentially private algorithms has proven surprisingly robust and useful.

One of these fundamental building blocks is {\em Differentially Private Selection}, which aims to select, based on a dataset, the best of many options. For concreteness, suppose that we have a score function $q : [K] \times \mathcal{D}^n \rightarrow \R$ that maps each of $K$ \emph{candidates}, and a dataset to a real-valued score. The DP selection problem is to select amongst these $K$ candidates, one that (approximately) maximizes this score on a given dataset $D \in \mathcal{D}^n$, while ensuring differential privacy.

One can only hope to approximately maximize $q$ when single individuals in the dataset cannot change any of the score functions $q(i, \cdot)$ too much. This stability of $q$ under small changes in $D$ is usually codified in an assumption that each score function $q(i, D)$ is Lipschitz with respect to Hamming distance $1$ changes to $D$. The Exponential mechanism~\cite{McSherryT} is an algorithm for DP selection under this assumption and has found numerous applications to the design of DP mechanisms. Several other mechanisms for the private selection problem have been proposed, that improve the utility guarantee under stronger assumptions~\cite{BeimelNS13, SmithT13, mir2013, ChaudhuriHS14, RaskhodnikovaS16, MinamiASN16}.

In many settings however, the Lipschitzness assumption is much too strong. In this work, we ask:
{\em Are there weaker versions of the stability assumption that allow for private selection?}
We show that one can codify the stability simply as differential privacy: the function $q$, viewed as a randomized algorithm, satisfies differential privacy.
Indeed, one can convert a Lipschitz function $q'$ into an $\eps$-DP random function $q$ by simply adding, say, a noise drawn from the Laplace distribution to $q'$.
We assume oracle access to a randomized function that on input $(i,\Data)$ computes a sample $(\xt,\qt)$ from the $i$-th candidate $\M_i(\Data)$, where $\qt$ is the score and $\xt$ can be any additional output.
Moreover, the output distributions of $\M_i( \Data)$ and $\M_i(\Data')$ are promised to be close whenever $\Data$ and $\Data'$ are neighbors. Here closeness in distributions is taken to mean $\eps$-DP or $(\eps, \delta)$-DP. Motivated by applications, we assume that 
the scores are bounded, say $\qt \in [0, 1]$.

To measure the quality of a candidate $\M_i(\Data)$, one option is by the \emph{median} of the distribution:
	$\Median\inp{\M_i(\Data)} \defeq \sup \set{\T \; : \; \Pr_{(\xt,\qt)\sim \M_i(\Data)}[\qt \ge \T] \ge \frac{1}{2}}$.
	However, even if a candidate $\M_i$ is $\eps_1$-DP, its median can still be very sensitive to the dataset.
	Thus one could only hope to approximately maximize the median score.
	Moreover, in many real world applications, one not only wants to find a ``good'' candidate, but also get a ``good'' sample from it,
	especially because these candidates themselves are randomized algorithms.
	Therefore, we use the following non-private algorithm as our main benchmark:
	draw a number of samples $(\xt_j,\qt_j)$ from every candidates, and then output the one with the highest score $\qt_j$.
	If one only assumes each candidate is individually $\eps_1$-DP, however, outputting the best of the $\widetilde{K}$ options will only be $\widetilde{K} \eps_1$-DP (see~\Cref{sec:naiveDP}).
	We would like to compete with this naive algortihm, while still preserving $O(\eps_1)$-DP.
Another important resource constraint in applications is the computational efficiency of the procedure. In our setting, we would want to minimize the number of oracle calls to $\M_i(\Data)$ made by our algorithm.

Our first result is a simple algorithm that given as input a threshold $\tau$, outputs a sample $(\xt,\qt)$ with score $\qt \ge \tau$, under the assumption that at least one candidate has a median score of at least $\tau$.
This algorithm makes a near linear number of oracle calls, and improves on the quadratic bound that follows from a reinterpretation of a result in~\cite{GLMRT}. We show that the loss in privacy, utility and efficiency for this algorithm are all close to optimal.
Interestingly, this algorithm can be seen as, starting from a naive differentially private algorithm with a poor utility guarantee (e.g., pick a candidate uniformly at random),
and then by repeating it in a private way to boost its utility guarantee.
In doing so, we get simple algorithms that are both private and have good utility guarantees.

\begin{theorem}
	Fix any $\eps_1>0, \tau \in [0,1]$.
  Then given $\eps_1$-DP algorithms $\M_1, \ldots, \M_K$, there is an algorithm $\M$ that on any dataset $D$, outputs a sample $(\xt,\qt)$ 
	such that 
	\begin{enumerate}[(a)]
	\item $\M$ is $(2\eps_1)$-DP. 
	\item $\qt\ge \tau$. 
	\item Let $\widetilde{T}$ be the number of calls the algorithm makes to any $\M_i(D)$, and suppose that $\exists i : \Pr_{q \sim \M_i(\Data)} \inb{q \ge \T} \geq \frac 1 2$, then
  $\E \widetilde{T} \le 2 K. 
  $
	\end{enumerate}
\end{theorem}


Can we do this without knowing this target value $\tau$? We give two algorithms that compete with the best $i$ without knowing the target $\tau$.
The first can be seen as modifying the naive non-private algorithm by employing a random stopping strategy.
In doing so, it guarantees that ``outputting the highest scored sample seen so far'' is already private.
However it pays a small additional privacy penalty: the final privacy cost is $3\eps_1$ instead of $2\eps_1$.

\begin{theorem}
	Fix any $\eps_1>0, \gamma \in [0,1]$.
  Then given $\eps_1$-DP algorithms $\M_1, \ldots, \M_K$, there is an algorithm $\M$ that on any dataset $D$, outputs a sample $(\xt,\qt)$ 
  such that 
	\begin{enumerate}[(a)]
  \item $\M$ is $(3\eps_1)$-DP.
  \item Let $\widetilde{T}$ be the number of calls the algorithm makes to any $\M_i(D)$, then
  $\E \widetilde{T} \le \frac{1}{\gamma} 
  $.
  \item $\qt$ is the highest scored sample among the $\widetilde{T}$ samples seen so far.
	  \end{enumerate}
\end{theorem}

Our second algorithm keeps the privacy cost to essentially $2\eps_1$, at the cost of a slightly higher runtime and a more complicated algorithm and analysis. This is valuable since in some settings, the utility of the base algorithm is quite sensitive with respect to the privacy parameter $\eps_1$. In such settings, with a final target privacy parameter of $\eps_{fin}$, the second algorithm can allow us to give each $\M_i$ a privacy budget of $\approx \eps_{fin}/2$, which can lead to a better utility than the $\approx (\eps_{fin}/3)$-DP $\M_i$'s needed for the first simpler algorithm.

\begin{theorem}
Fix any $\eps_1>0, \eps_0 \in [0,1], \beta > 0, R \in \mathbb{N}$.
Suppose that there are $\eps_1$-DP algorithms $\M_1, \ldots, \M_K$ and let $\tau^*(D) = \max_i \Median(\M_i(D))$. There is an algorithm $\M$ that on any dataset $D$ either outputs $\bot$, or outputs a sample $(\xt,\qt)$ 
such that
\begin{enumerate}[(a)]
	\item 
		$\M$ is $(2\eps_1+\eps_0, \delta)$-DP.
	\item 
		Except with probability $\beta+\delta/R$, $\xt$ has quality at least  $\tau^* - \frac{1}{R}$.
	\item 
The number of calls $\widetilde{T}$ that the algorithm makes to any $\M_i(D)$ satisfies (deterministically)
\begin{align*}
	\widetilde{T} \le O\inp{K\inp{\frac{R+1}{\beta^2}}^{6+\frac{12\eps_1}{\eps_0}} \inp{\frac{\ln\frac{R}{\delta}}{\eps_0^2} + \frac{\ln \frac{1}{\eps_0}}{\beta}} }.
\end{align*}
  Furthermore,
	$\Pr\inb{\M \hbox{ outputs $\bot$} } \le \beta+\delta$.
\end{enumerate}
\end{theorem}

In the process, we develop an online version of our algorithm, which can be seen as a generalization of the sparse vector technique~\cite{DworkNRRV09} to this privacy-instead-of-Lipschitzness setting. 
This algorithm takes as input a sequence of mechanisms $\M_i(\cdot)$ and $\T_i$, and stops at the first $i$ such that $\M_i(\cdot)$ has median score larger than $\T_i$.

\begin{theorem}
	There is an $(\eps_3,\delta)$-DP mechanism $\M_{sv}$ such that,
	for any $p^* \in (0,1), \beta\in (0,1)$,
	and for any sequence of $\eps_1$-DP mechanisms $\M_1, \cdots, \M_K$ and any sequence of thresholds $\T_1, \cdots, \T_k$: 
	\begin{enumerate}[(a)]
		\item 
	  If there is an $i$ such that $\Pr[\M_i(D) \geq  \T_i] \geq p^*$, then $\M$ outputs $i$ with probability $(1-\beta)$. 
  \item 
	 If $\M$ outputs $i$, then except with probability $\beta$, $\Pr[\M_i(D) \geq \T_i] \geq (\frac{\beta}{K})^{O(\eps_1/\eps_3)} \cdot p^*$.
	\end{enumerate}
\end{theorem}

Several remarks are in order. First note that the stability assumption that we use, i.e. that of differential privacy, is in some sense the weakest possible. Indeed if we want the final outcome to be differentially private and we treat each mechanism as a blackbox, it is easy to see that each mechanism itself must be differentially private. In other words, we have relaxed the Lipschitzness condition to the weakest possible condition that would allow for differentially private selection.  Our algorithm suffers a factor of two loss in the privacy parameter. In Appendix~\ref{app:lower_bounds}, we show that this factor of two loss in unavoidable even in simple settings. Note also that our algorithm only makes $\tilde{O}(K)$ oracle calls, whereas even computing the maximum non-privately would require $K$ oracle calls.

We next outline some motivating applications of our work.

\medskip\noindent{\bf Hyperparameter/algorithm Selection:} When designing practical machine learning algorithms, one often ends up choosing amongst different algorithms/models, or setting values for common hyperparameters such as the learning rate in an algorithm. This {\em hyperparameter selection} problem has attracted a lot of interest in recent years~\cite{wiki:hyperparameter}. Differentially private ML algorithms such as~\cite{abadi2016deep, papernot2016semi} have many of these hyperparameters, and often add on a few hyperparameters of their own. A common approach in the non-private setting is to try out several (or all) values of the hyperparameters and select the best one based on the performance on a validation set. Doing this with privacy requires more care. Chaudhuri and Vinterbo~\cite{ChaudhuriV13} studied this problem formally under strong assumptions on the algorithm. These assumptions, however, can be hard to enforce and one would like to design an algorithm that works without any additional assumptions. Note that given $K$ choices for the hyperparameters, and an $\eps$-DP
learner, one can publish $K$ models and select the best, say using the exponential mechanism. This approach only gives $\eps K$-DP, which allows for  privacy budget of only $\eps/K$ (or $\eps\sqrt{\log \frac 1 \delta}/\sqrt{K}$ if using advanced composition) for the learner, which often translates to significantly poorer utility guarantee. In this setting, note also that each oracle call is a run of the DP learner for some hyperparameter setting, that can involve a large computational cost.

Our work shows how to compete with the best choices of hyperparameters in the non-private setting while satisfying $O(\eps)$-DP, at a small computational overhead.

\medskip\noindent{\bf Adaptive data analysis beyond low-sensitivity queries:} One of the applications of DP, beyond privacy itself, is in understanding overfitting in the adaptive setting where the same dataset is used in a sequence of analyses, chosen adaptively based on the results of previous ones. This problem, sometimes referred to as the {\em garden of forking paths}~\cite{gelman2014statistical}, can lead to a breakdown of standard statistical guarantees. A beautiful recent line of work~\cite{DworkFHPRR15, BassilyNSSSU16} shows that when these analyses take the form of low-sensitivity queries, using differentially private versions of these analyses allows us to improve the sample complexity quadratically over what would otherwise be possible.  Often, however, the forking paths can involve queries that are not low-sensitivity. For example, at some step an analyst may choose the best $k$ for $k$-means clustering or may choose the clustering algorithm itself amongst one of several. At another step, the analyst may project the data for a carefully chosen target rank, and may choose to use a projection algorithm such as PCA, or an $\ell_p$ version of PCA to get outlier robustness, for a carefully chosen $p$. Making these choices differentially private naively would involve paying for the privacy cost of each of the options considered, even though only one may be used in the subsequent analysis. Our work shows that if one uses a differentially private algorithm to score each of the options, selecting amongst them can be done while paying the adaptivity cost of only one query, essentially independently of the number $K$ of options considered.

\medskip\noindent{\bf Generalizing the Exponential Mechanism: } Beyond these applications, our result can be viewed as a generalization of the expoenential mechanism. 
Given a score fuction $q$ that has sensitivity $S$, observe that adding Laplace noise of scale $S/\eps$ to the score gives us an $\eps$-DP mechanism. 
Our algorithm can be then used to select amongst these. We can however relax the assumptions. If we allow the score functions to have different sensitivities, we can still use our framework and recover the generalized exponential mechanism of Raskhodnikova and Smith~\cite{RaskhodnikovaS16}. If the score functions have small smoothed sensitivity~\cite{NissimRS07}, we get a smooth sensitivity version of the exponential mechanism. This last result does not seem to follow from known techniques.

\medskip\noindent{\bf Private amplification for private algorithms: } Beyond these applications, our result can be viewed as an extension of the \emph{private amplification scheme} introduced in~\cite{GLMRT}.
Given a private algorithm, which is usually a randomized algorithm, ideally one would like to run it multiple times, and then choose the \emph{best} run so as to obtain an output with a higher quality.
Here the quality measure can either be the success probability, or any other utility measure of the output.
This is trivial in the non-private setting. Is it possible to compete with such a naive repetition strategy in a differentially private way?
In this work, we present an algorithm that can be seen as modifying the naive repetition strategy with a random stopping time, which is arguably almost as competitive as the non-private naive repetition.

\subsection{Other Related Work}
The Differentially Private Selection problem, often known as differentially private maximization, is a very general algorithmic problem that arises in many applications. Some examples include private PAC learning~\cite{KLNRS}, private
frequent itemset mining~\cite{BhaskarLST10}, private PCA~\cite{ChaudhuriSS12, KapralovT13} and private multiple hypothesis testing~\cite{UhleropSF13, DworkSZ15}. The Sparse Vector Technique can be viewed in hindsight as a novel solution to the online version of the selection problem, under the assumption that the target value $\tau$ is known in advance. This technique was introduced by Dwork et al.~\cite{DworkNRRV09}.
We refer the reader to the book by Dwork and Roth~\cite{dwork2014algorithmic} for further applications of these techniques.

Several generalization of the exponential mechanims have been proposed. Smith and Thakurta~\cite{SmithT13} and Beimel et al.~\cite{BeimelNS13} showed that the utility guarantee can be improved using the propose-test-release framework of Dwork and Lei~\cite{DworkL09} when there is a large margin between the maximum and the rest. Chaudhuri et al.~\cite{ChaudhuriHS14} gave an elegant algorithm that can exploit a large margin between the maximum and the $k$th maximum for any $k$. Raskhodnikova and Smith~\cite{RaskhodnikovaS16} proposed the generalized exponential mechanism whose utility depends on the sensitivity of the maximizer, rather than the worst-case sensitivity. Minami et al.~\cite{MinamiASN16} show that under certain assumptions on the base distribution, the sensitivity assumptions on the loss function can be significantly relaxed.
Our algorithms can also be seen as a natural generalization of the Laplace mechanism.
Given a Lipschitz score function $q'$, one can convert it into an $\eps$-DP score function $q$ by adding a Laplace noise. 
Then the Laplace mechanism says that one can just output the max of the noise-added scores. 
However, the Laplace relies crucially on the fact that the noise is a Laplace noise. 
As we will discuss in~\Cref{sec:naiveDP}, under the mere assumption that the score function is $\eps$-DP, outputting the max will inevitably incur a factor of $K$ loss in privacy.

The problem of algorithm selection has also been studied in~\cite{pythia} where the best parameters are learnt from features of the problem. Ligett et al.~\cite{LigettNRWW17} study the problem of picking from a sequence of algorithms with increasing privacy costs, until one with good utility is found, for a special class of mechanisms.

The problem of private median finding, and more generally private percentile estimation has been studied in several works~\cite{NissimRS07,DworkL09, Smith11, BunNSV15}. While syntactically similar to the threshold estimation problem studied in Section~\ref{sec:sparse_vector}, the assumptions on the data in those works are very different from ours and we do not believe that the techniques in those works apply to the setting of interest in this work.

\subsection{Organization}
The rest of the paper is organized as follows. In Section~\ref{sec:priv_select} we present our algorithm for the known threshold case. Section~\ref{sec:sparse_vector} describes our sparse vector and general selection algorithms. We sketch applications of our results in Section~\ref{sec:applications}. The appendices contain some deferred proofs, show why simpler natural approaches do not work for our problem, and show a lower bound on the privacy overhead.

%% file: prelims.tex
\section{Preliminary and Notations}
For a random variable $X$ and distribution $Q$, we write $X \sim Q$ if $X$ is distributed according to the law of $Q$.


Let $\dis{A}{B}$ be the max-divergence of two random variables defined as follows:
\[
  \dis{A}{B} = \max_{S \subseteq \supp(A)} \inb{\ln \frac{\Pr[A \in S]}{\Pr[B \in S]}}.
\]
Then we define $\disd{A}{B}$ as:
\[
  \disd{A}{B} = \max_{S \subseteq \supp(A)} \inb{\ln \frac{\Pr[A \in S]-\delta}{\Pr[B \in S]}}.
\]
For convenience, for distributions $Q_1$ and $Q_2$, let $q_1,q_2$ be random variables distributed as $Q_1$ and $Q_2$ respectively, then we will also write $\dis{Q_1}{Q_2} \defeq \dis{q_1}{q_2}$, and similarly $\disd{Q_1}{Q_2} \defeq \disd{q_1}{q_2}$.

For a distribution $Q(\Data)$ that depends on datasets $\Data$, we say that $Q$ satisfies \emph{$\eps$-differential privacy} (or simply written as $\eps$-DP), if for every two neighboring datasets $\Data_1, \Data_2$, $\dis{Q(\Data_1)}{Q(\Data_2)} \le \eps$.
And we say $Q$ satisfies \emph{$(\eps,\delta)$-DP} if for every two neighboring datasets $\Data_1, \Data_2$, $\disd{Q(\Data_1)}{Q(\Data_2)} \le \eps$.

Given a function $f$ on dataset $\Data$, we say that $f$ is $t$-Lipschitz if for any two neighboring dataset $\Data,\Data'$, $\abs{f(\Data) - f(\Data')} \le t$.

%% file: priv_select.tex
\section{Private selection}
\label{sec:priv_select}

Let $\set{M_i(\Data)}_{i=1}^K$ be a set of differentially private mechanisms, that is, for every $i$, $M_i$ is a differentially private mechanism with respect to the dataset $\Data$.
We will also refer to the set of $M_i$ as \emph{private candidates}.
For convenience, we will also treat a randomized mechanism $M_i(\Data)$ as a distribution, and write $m \sim M_i(\Data)$ if $m$ follows the output distribution of $M_i(\Data)$. 
Let $\set{q_i}$ be scoring functions over the output of these mechanisms, that is, for $m \sim M_i(\Data)$, $q_i(m)$ is the score for $m$.
We assume that there is a total ordering of the candidates: when two candidates have the same score, we assume that there is an arbitrary tie-breaking rule (e.g., by alphabetical ordering). Given a total ordering of the candidates, without loss of generality we will further assume that each option has a different score.


The goal of \emph{private selection} is to select $(m,i)$ that (approximately) maximizes the score of $q_i(m)$.
Naively, a natural algorithm is to draw samples $m_i \sim M_i$ for every $i$, and then output the pair $(m_i,i)$ with the highest score $q_i(m_i)$.
Unfortunately this naive algorithm is not private.
The detailed discussion and analysis is deferred to~\Cref{sec:naiveDP}.
The next natural algorithm would be to output the $p$-th percentile best, which unfortunately is also not private. Again we defer the analysis to~\Cref{sec:percentileDP}.

In this section, we will start with the following naive algorithm that is guaranteed to be \emph{private} but not very \emph{useful} (has poor utility guarantee):
we choose a candidate $i$ uniformly at random and output $M_i(D)$.
It is not hard to see that such a choice of candidate is at least as private as the individual candidates.
However, the probability of getting a reasonably ``good'' candidate can be of the order $O(1/K)$.
Nevertheless, we will show how to boost its usefulness (utility guarantee) by \emph{thresholding} or \emph{random stopping}.
As a result, this leads to simple and practical algorithms that are also able to compete with the \emph{best} candidates in a differentially private way.

Formally, we consider a randomized mechanism $Q(\Data)$, where every output comes with a utility score $q$: $(x,q) \in \Omega \times \R$.
For convenience, we will abuse notation and also denote the output distribution of the randomized mechanism $Q(\Data)$ by $Q(\Data)$, and write $(\xt,\qt) \sim Q(\Data)$ to indicate that $(\xt,\qt)$ is obtained by running the randomized mechanism $Q(\Data)$.
Given blackbox access to $Q(\Data)$, the goal is to find $(x,q)$ that (approximately) \emph{maximizes} the score: e.g., they are the top $1\%$, that is, $\Pr_{(\xt,\qt)\sim Q(\Data)} [\qt > q] < 0.01$.
When it is clear from the context, we will also simply write $q \sim Q(\Data)$ for taking only the $q$ part of the pair $(x,q)$.

To apply this framework to the private selection problem, we define a randomized mechanism $Q(\Data)$ as follows:
we first sample $i \sim \mathrm{Uniform}[K]$, then sample $m \sim M_i(\Data)$, and evaluate the score $q_i(m)$, and output $\left( (i,m),q_i(m) \right)$.
The above sampling process 
implements an oracle access to the naive algorithm that outputs a candidate uniformly at random.
Our goal is to boost the utility of such a naive algorithm.
While our algorithms work for more general distribution of $Q(\Data)$, where the candidate $i$ can be drawn from \emph{any} samplable distribution,
we will focus in this work on the case when $i$ is drawn uniformly from a finite set of candidates (e.g., due to the lack of domain knowledge).
It is also worth noting that if $M_i$ is $\eps_1$-DP for every $i$, then so is $Q$.
Similarly if $M_i$ is $(\eps_1,\delta_1)$-DP for every $i$, then so is $Q$.

\subsection{Private selection with a known threshold $\T$}
We consider a thresholding algorithm,
which
for a given threshold,
repeatedly samples from the candidates until we get one that is above the threshold. In addition, we have a small probability $\gamma$ of stopping at each step.
See~\cref{alg:thresholding} for a more formal description.

  \begin{algorithm}
    \flushleft Input: a threshold $\T$, a budget $\gamma \le 1$ and $\eps_0 \le 1$, number of steps $T\ge \max\set{\frac{1}{\gamma}\ln \frac{2}{\eps_0 } , 1+\frac{1}{e\gamma} }$, and sampling access to $Q(\Data)$.

    For $j = 1, \cdots, T$:
  \begin{itemize}
    \item draw $(x,q) \sim Q(\Data)$
    \item if $q \ge \T$ then output $(x,q)$ and halt; 
    \item flip a $\gamma$-biased coin: with probability $\gamma$, output $\bot$ and halt;
  \end{itemize}
  Output $\bot$ and halt.
  \caption{Thresholding with a known threshold $\T$.}
  \label{alg:thresholding}
\end{algorithm}

We assume that the adversary can only observe the final output of the algorithm.
We show that for any choice of parameters, the algorithm is private; and if the given threshold $\T$ is a ``good'' threshold, the algorithm is unlikely to output $\bot$.
\begin{theorem}
  Fix any $\eps_1, \delta_1 >0, \eps_0 \in [0,1], \gamma \in [0,1]$.
  Let $T$ be any integer such that $T\ge \max\set{\frac{1}{\gamma}\ln \frac{2}{\eps_0 } , 1+\frac{1}{e\gamma} }$,
  Then~\cref{alg:thresholding} with these parameters satisfies the following:
  \begin{enumerate}[(a)]
	  \item Let $\Aout(\Data)$ be the output of~\cref{alg:thresholding}, then for $q\geq\tau$,
		  \[\Pr[\Aout(\Data) = (x,q)] \propto \Pr_{(\xt,\qt) \sim Q(\Data)}[(\xt,\qt) = (x,q)].\]
	  \item If $Q$ is $\eps_1$-DP, then the output is $(2\eps_1 + \eps_0)$-DP.
	  \item If $Q$ is $(\eps_1,\delta_1)$-DP, then the output is $\inp{2\eps_1 + \eps_0, \; 3 e^{2\eps_1 + \eps_0}\cdot\frac{\delta_1}{\gamma}}$-DP.
    \item Let $\widetilde{T}$ be the number of iterations of the algorithm, and let $p_1 = \Pr_{q \sim Q(\Data)} \inb{q \ge \T}$, then
	    \[\E \widetilde{T} \le \frac{1}{p_1(1-\gamma) + \gamma} \le \min\set{\frac{1}{p_1}, \frac{1}{\gamma}}.\]
    \item
  Furthermore,
  $\Pr\inb{ \hbox{output $\bot$} } \le \frac{(1-p_1) (1+ \eps_0/2)}{p_1}\gamma$.
  \end{enumerate}
  \label{thm:thresholdingDP}
\end{theorem}
  Due to space considerations, we defer this proof to~\Cref{sec:proof-thresholding}.
  As a remark,  it is clear that in the worst case, the number of iterations of~\cref{alg:thresholding} is no more than $T$; this theorem provides a more average-case guarantee: the larger $p_1$ (or $\gamma$) is, the more likely that the algorithm will terminate (much) sooner than $T$.
Moreover, it is worth noting that the larger the setting of $T$ is, the smaller we can set $\eps_0, \gamma$, providing more privacy and utility.
In particular, the above theorem holds even if we set $\gamma=0,\eps_0=0$ but $T=\infty$, in other words, we run the algorithm till it stops by itself.
However this would not be a very practical setting: if one started with a ``bad'' threshold, the algorithm may never stop. In that case, one may want to stop the algorithm and try a different threshold.
Therefore, for practical purposes one may want to set $\gamma>0$ and $\eps_0>0$.

\subsection{Random stopping without thresholding}
In this subsection, we show that the idea of random stopping leads to a simple private algorithm, even without knowing the threshold.
It is similar to~\cref{alg:thresholding} but without the thresholding part: draw a random number of samples, and then output the best option.

  \begin{algorithm}
    \flushleft Input: a budget $\gamma \le 1$ and the sampling access to $Q(\Data)$.

    Initialize the list (multiset) $S=\emptyset$.

    For $j = 1, \cdots, \infty$:
  \begin{itemize}
    \item draw $(x,q) \sim Q(\Data)$
    \item $S \gets S \cup \set{(x,q)}$
    \item flip a $\gamma$-biased coin: with probability $\gamma$, we output the highest scored candidate from $S$ and halt;
  \end{itemize}
  \caption{Outputting the highest score with random stopping.}
  \label{alg:maxRand}
\end{algorithm}

\begin{theorem}
	Fix any $\eps_1>0, \gamma \in [0,1]$.
	    If $Q$ is $\eps_1$-DP, then the output of~\cref{alg:maxRand} is $(3\eps_1 )$-DP.
  \label{thm:maxRandDP}
\end{theorem}

\begin{proof}
	We first consider the event of getting the output $(x,q)$ from~\cref{alg:maxRand} on neighboring datasets $\Data$ and $\Data'$.
  Without loss of generality, we assume that each option has a different score.\footnote{Otherwise, whenever we write $q_1 > q$, we break ties using the same total ordering of the candidates.}
  Then we denote
  \begin{align*}
	  p \defeq& \Pr_{\qt \sim Q(\Data)} \inb{\qt = q } \quad \quad \hbox{ and } \quad \quad p' \defeq \Pr_{\qt \sim Q(\Data')} \inb{\qt = q }, \\
	  p_0 \defeq& \Pr_{\qt \sim Q(\Data)} \inb{\qt > q } \quad \quad \hbox{ and } \quad \quad p_0' \defeq \Pr_{\qt \sim Q(\Data')} \inb{\qt > q}, \\
	p_1 \defeq& \Pr_{\qt \sim Q(\Data)} \inb{\qt \ge q } \quad \quad \hbox{ and } \quad \quad p_1' \defeq \Pr_{\qt \sim Q(\Data')} \inb{\qt \ge q} .
  \end{align*}
  Notice that $p= \Pr_{(\xt,\qt) \sim Q(\Data)} \inb{ (\xt,\qt) = (x,q) }$, $p_1= p_0 + p$, and $p_1' = p_0' + p'$.

  We define the highest score for a set (or a multiset) $S$ of tuples $(x,q)$ as
  \[
	  \max S\defeq \max_{ (x,q) \in S } q. 
  \]
  Let $\Aout(\Data)$ be the output of~\cref{alg:maxRand} on $\Data$,  then we have
  \begin{align*}
	  &\Pr\inb{ \Aout(\Data) = (x,q) } \\
	  =& \sum_{j=1}^{\infty}\Pr\inb{ \Aout(\Data) = (x,q) \wedge \abs{S} = j } \\
	  =&\sum_{j=1}^\infty \Pr\inb{\abs{S} = j} \cdot \Pr\inb{\hbox{$\max S \le q$, and $(x,q) \in S$ }\mid \abs{S} = j} \\
	  =&\sum_{j=1}^\infty \inp{1-\gamma}^{j-1} \gamma \cdot \Pr\inb{\hbox{$\max S \le q$, and $(x,q) \in S$ }\mid \abs{S} = j}.
  \end{align*}
  Then, observe that
  \begin{align*}
	  \Pr\inb{\hbox{$\max S \le q$} \mid \abs{S}=j} = (1-p_0)^j,
  \end{align*}
  and
  \begin{align*}
	  \Pr\inb{ (x,q) \in S \mid \hbox{$\max S \le q$, and $\abs{S}=j$}} = 1 - \inp{1-\frac{p}{1-p_0}}^j= 1 - \inp{\frac{1-p_1}{1-p_0}}^j.
  \end{align*}
  Together we have
  \begin{align*}
    \Pr\inb{ \Aout(\Data) = (x,q) }
    =&\sum_{j=1}^\infty \inp{1-\gamma}^{j-1} \gamma \cdot (1-p_0)^j \inp{1 - \inp{\frac{1-p_1}{1-p_0}}^j} \\
    =&\sum_{j=1}^\infty \inp{1-\gamma}^{j-1} \gamma \cdot  \inp{(1-p_0)^j - (1-p_1)^j} \\
    =&\frac{\gamma(1-p_0)}{1 - (1-\gamma)(1-p_0)} - \frac{\gamma(1-p_1)}{1 - (1-\gamma)(1-p_1)} \\
    =&\frac{\gamma(p_1-p_0)}{\inp{p_0(1-\gamma)+\gamma}\inp{p_1(1-\gamma)+\gamma}} \\
    =&\frac{\gamma p}{\inp{p_0(1-\gamma)+\gamma}\inp{p_1(1-\gamma)+\gamma}} .
  \end{align*}
  Since $Q$ is $\eps_1$-DP, we have that $p,p_0,p_1$ are $\eps_1$-close (in a DP sense) to $p',p_0',p_1'$, respectively. 
  Then,
  \begin{align*}
    \frac{\Pr\inb{ \Aout(\Data) = (x,q) } }{\Pr\inb{ \Aout(\Data') = (x,q) } }
    = \frac{p}{p'}\cdot \frac{p_0'(1-\gamma) + \gamma}{p_0(1-\gamma) + \gamma} \cdot \frac{p_1'(1-\gamma) + \gamma}{p_1(1-\gamma) + \gamma}
    \le  \exp(3\eps_1) .
  \end{align*}
\end{proof}
The following utility bound holds for this algorithm.
\begin{theorem}
  For $p > 0$, let $\Qp(D) = \sup \{z: \Pr[Q(D) \ge z] > p\}$. Then the output of~\cref{alg:maxRand} has score at least $\Qp(D)$ except with probability $\gamma / p$.
\end{theorem}
\begin{proof}
Let $\Aout(\Data)$ be the output of~\cref{alg:maxRand} on $\Data$, then we write
\begin{align*}
	\Pr[\Aout(\Data) < \Qp(\Data)]
  & = \sum_{j=1}^\infty \Pr\inb{\abs{S} = j} \cdot \Pr\inb{\max S < \Qp(\Data)\mid \abs{S} = j}\\
  &= \sum_{j=1}^\infty \inp{1-\gamma}^{j-1} \gamma \cdot  (1-p)^j\\
  &\leq \gamma / p.
\end{align*}
\end{proof}
Instead of random stopping, one can also design a hard stopping variant of this algorithm similar to that of~\cref{alg:thresholding}, and allow for $\ed$-DP input algorithms.

\begin{theorem}
  	Fix any $\gamma \in [0,1], \delta_2 > 0$ and let $T=\frac 1 \gamma \log \frac 1 {\delta_2}$. Consider a variant of~\cref{alg:maxRand} that outputs the highest scored candidate from $S$ if $j$ reaches $T$.
  	    If $Q$ is $(\eps_1, \delta_1)$-DP, then the output of this algorithm is $(3\eps_1 + 3\sqrt{2\delta_1}, \delta)$-DP for $\delta = \sqrt{2\delta_1} T + \delta_2$.
    \label{thm:maxRandDPapprox}
\end{theorem}
\begin{proof}
  We simply reduce to~\cref{thm:maxRandDP} using simple properties of $(\eps, \delta)$-DP. We give details next, using folklore results proven in~\Cref{app:ed_v_cond_e}. Fix a pair of neighboring datasets $D$ and $D'$. Then we can define an event $B$ such that $\Pr[B] \leq \sqrt{\delta_1}$ and that $Q(D) \mid B^c$ and $Q(D') \mid B^c$ are multiplicatively $\eps_1 + \sqrt{2\delta_1}$ close. Let $B_j$ be the event $B$ in the $j$th call to $Q$.
   Further, let $C$ be the event that the algorithm reaches step $T$. Conditioned on $(\cup_{j=1}^T B_j \cup  C)^c$, the run of this algorithm can be coupled with a run of~\cref{alg:maxRand} for a pure DP $Q$. Further, the probability of the event $\cup_j B_j \cup C$ is at most $\sqrt{2\delta_1} T + \delta_2$. The claim follows.
\end{proof}
Since $\delta_1$ is typically smaller than a polynomial, we have not attempted to optimize the $\delta$ term in this theorem.
We conclude with a remark that, in the case when $Q$ satisfies purely $\eps_1$-DP, one can show that the hard stopping variant of~\cref{alg:maxRand} preserves purely $\approx 3\eps_1$-DP.

\begin{theorem}
	Fix any $\eps_0 \in (0, 1/2), \gamma \in [0,1], \delta_2 > 0$ and let $T=\left\lceil \frac{1}{\gamma}  \inp{\ln \frac{2(1+\gamma)^2}{\eps_0 \gamma^2} + \ln \ln \frac{2(1+\gamma)^2}{\eps_0 \gamma^2}}\right\rceil$. Consider a variant of~\cref{alg:maxRand} that outputs the highest scored candidate from $S$ if $j$ reaches $T$.
	If $Q$ is $\eps_1$-DP, then the output of this algorithm is $(3\eps_1 + 3\eps_0)$-DP. 
    \label{thm:maxRandDPstop}
\end{theorem}
The proof of this theorem is quite involved and is deferred to~\Cref{sec:proof-maxRandDPstop}.

%% file: general_sparse_vector.tex
\section{Searching for a percentile-threshold: privacy-preserving sparse vector}
\label{sec:sparse_vector}
In this section we consider the problem of searching for a percentile-threshold $\T$ for any given percentile $\pstar$ in a differentially private way.
We start by defining some notations.
Given any sequence of randomized queries $\set{Q_i}$, 
we write $q_i \sim Q_i(\Data)$ to indicate that $q_i$ is obtained from running the randomized query $Q_i$ on dataset $\Data$.
In other words, $q_i \sim Q_i(\Data)$ means that $q_i$ follows the output distribution of the randomized query $Q_i$ on dataset $\Data$.
We will treat these $Q_i(\Data)$ as samplable distributions, where each $Q_i$ is $\eps_1$-DP.
Then for any sequence of thresholds $\set{\T_i}$,
and a target threshold $\pstar \in (0,1)$,
we would like to test if $\Pr_{q_i \sim Q_i(\Data)}[q_i \ge \T_i] > \pstar$ 
and output the first one that is above the threshold, and in a differentially private way. 

It is worth noting that this can be seen as an extension of the standard sparse vector algorithm for Lipschitz queries:
given $1$-Lipschitz queries $f_1,\cdots, f_k$ and a threshold $\T_0$, if we set $\pstar=\frac{1}{2}$, $Q_i = f_i + \Lap\inp{\frac{4}{\eps_1}}$ and $\T_i = \T_0$,
then it is not hard to check that the queries $Q_i$ are now $\eps_1$-DP, and the first query $Q_i$ above the percentile-threshold is exactly the same as the first query $f_i$ above the query threshold $\T_0$ (that is, the first $f_i$ with median score at least $\T_0$).
Answering such a percentile query exactly is not private (see~\Cref{sec:percentileDP} for an example for $\pstar=1/2$),
so we will have to relax the goal of finding the first above percentile-threshold query.
Similar to the standard setting, we would like that:
\begin{itemize}
	\item
if a query is much below the threshold, that is, $\Pr_{q_i \sim Q_i(\Data)}[q_i \ge \T_i] \ll \pstar$, then our algorithm should report ``below threshold'' (denoted by $\bot$);
\item
if a query is much above the threshold, that is, $\inp{1- \Pr_{q_i \sim Q_i(\Data)}[q_i \ge \T_i]} \ll \inp{1-\pstar}$, then our algorithm should report ``above threshold'' (denoted by $\top$).
\end{itemize}
In fact, our algorithm will be a natural extension of the standard sparse vector algorithm.

\subsection{Sparse vector for online private queries with the help of a percentile oracle}
To illustrate ideas, we will start by assuming that we have access to an \emph{exact} percentile oracle: $\p(\T_i,Q_i) \defeq \Pr_{q_i\sim Q_i} [q_i \ge \T_i]$.
As a remark, such a percentile oracle is available in the standard sparse vector algorithm, which is just the cumulative distribution function of the Laplace distribution.
We observe that if the randomized queries $Q_i$ are $\eps_1$-DP, then both $\ln \p(\T_i,Q_i)$ and $\ln \inp{1-\p(\T_i,Q_i)}$ are $\eps_1$-Lipschitz.
In other words, although we no longer have Lipschitzness in the ``answer of a query'' space (that is, the quantile space), the fact that each query is $\eps_1$-DP will ensure that we have Lipschitzness in the logarithm of the percentile space (that is, the log of the CDF). This allows us to adapt the sparse vector algorithm to the log of the percentile space.

Let $\Phi(x) \defeq \frac{x}{1-x}$.
Note that $\ln \Phi\inp{\p(\T_i,Q_i)}$  is $2\eps_1$-Lipschitz: since both $\ln \p(\T_i,Q_i)$ and $\ln \inp{1-\p(\T_i,Q_i)}$ are $\eps_1$-Lipschitz,  and $\ln \Phi\inp{\p(\T_i,Q_i)}$ is just the difference of two $\eps_1$-Lipschitz functions.
Also notice that $\Phi$ is a strictly increasing function for $x \in (0,1)$.
Given access to the oracle $\p(\T_i,Q_i)$, we can then adapt the sparse vector algorithm as in~\cref{alg:gensparse0}.

\begin{algorithm}
	\flushleft Input: $\eps_1, \eps_3, \pstar$, a stream of thresholds $\set{\T_i}$ and randomized queries $\set{Q_i(\Data)}$ 

  Sample $\nu \sim \Lap\inp{\frac{4\eps_1}{\eps_3}}$;

  For $i=1, \cdots $:
  \begin{itemize}
    \item let $\xi_i \sim \Lap\inp{\frac{8\eps_1}{\eps_3}}$;
    \item $e^{\xi_i}\cdot \Phi\inp{\p(\T_i, Q_i)} > e^{\nu} \Phi(\pstar)$,
	    output $a_i = \top$ and halt;
    \item otherwise output $a_i=\bot$;
  \end{itemize}
  \caption{AboveThreshold algorithm assuming the oracle $p(\T_i, Q_i)$.}
  \label{alg:gensparse0}
\end{algorithm}

\begin{theorem}
If for every $i$, $Q_i$ is $\eps_1$-DP,
  then 
  \begin{enumerate}[(a)]
    \item \cref{alg:gensparse0} is $\eps_3$-DP.
    \item Conditional on \cref{alg:gensparse0} reporting the $R$-th query $Q_R$ is ``above threshold'',
	    we have that $\forall \beta \in (0,1),
	    \Pr\inb{\Phi\inp{\p(\T_R, Q_R)} \le \inp{\frac{\beta}{R+1}}^{\frac{12\eps_1}{\eps_3}}  \Phi(\pstar)} \le \beta
      $.
      In other words, the algorithm does not stop too early.
    \item Conditional on \cref{alg:gensparse0} reporting the $R$-th query $Q_R$ is ``above threshold'',
	    we have that $\forall \beta \in (0,1),
	    \Pr\inb{\exists i < R : \Phi\inp{\p(\T_i, Q_i)} \ge \inp{\frac{R+1}{\beta}}^{\frac{12\eps_1}{\eps_3}}  \Phi(\pstar)} \le \beta
      $.
      In other words, the algorithm does not stop too late.
    \item $\forall \beta \in (0,1)$, if for some $i$,
	    $
	\Phi\inp{\p(\T_i, Q_i)} \ge \inp{\frac{1}{\beta}}^{\frac{12\eps_1}{\eps_3}}  \Phi(\pstar),
	$
	then
	$\Pr\inb{a_i = \top | \forall j<i, a_j = \bot} \le \beta$.
	In other words, on a query that is way above the threshold the algorithm will likely halt.
  \end{enumerate}
  \label{thm:gensparse0}
\end{theorem}

\begin{proof}(Sketch)
	Part (a), part (b) and part (c) all follow from the standard sparse vector analysis (see, e.g., \cite{dwork2014algorithmic}), and the fact that $\ln \Phi\inp{\p(\T_i,Q_i)}$ is $2\eps_1$-Lipschitz.
	Observe that the test $e^{\xi_i}\cdot \Phi\inp{\p(\T_i, Q_i)} > e^{\nu} \Phi(\pstar)$ is equivalent to $\xi_i + \ln \Phi\inp{\p(\T_i,Q_i)} > \nu + \ln \Phi(\pstar)$.
	Therefore, if we view $\ln \Phi\inp{\p(\T_i,Q_i)}$ as the $i$-th query (which is $2\eps_1$-Lipschitz) and $\ln \Phi(\pstar)$ as the threshold, then this is indeed the standard sparse vector setting.
	The details are omitted here as we will see proofs for stronger claims for the actual algorithm in~\cref{thm:gensparse1}. 

	For part (d), 
	observe that the test $e^{\xi_i}\cdot \Phi\inp{\p(\T_i, Q_i)} > e^{\nu} \Phi(\pstar)$ will pass if
	$\xi_i \ge -\frac{8\eps_1}{\eps_3} \ln \frac{1}{\beta}$
	and $\nu \le \frac{4\eps_1}{\eps_3} \ln \frac{1}{\beta}$.
	By a union bound, with probability at least $1-\beta$, both will happen at the same time.
	In other words, the probability of not halting after seeing a query way above the threshold is at most $\beta$.
\end{proof}

We give some estimates in the special case of $\pstar=1/2$, which corresponds to the range of the standard sparse vector setting, as quick corollaries.
In fact, if one apply this to the standard sparse vector setting, one can recover guarantees that match the standard setting up to constant factors.
\begin{corollary}
	If $\pstar = 1/2$, and for every $i$, $Q_i$ is $\eps_1$-DP,
  then 
  \begin{enumerate}[(a)]
    \item \cref{alg:gensparse0} is $\eps_3$-DP.
    \item Conditional on \cref{alg:gensparse0} reporting the $R$-th query $Q_R$ is ``above threshold'',
	    we have that $\forall \beta \in (0,1),
	    \Pr\inb{\p(\T_R, Q_R) \le \frac{\beta^{\frac{12\eps_1}{\eps_3}}}{\beta^{\frac{12\eps_1}{\eps_3}} + (R+1)^{\frac{12\eps_1}{\eps_3}}}  } \le \beta
      $.
      In other words, the algorithm does not stop too early.
    \item Conditional on \cref{alg:gensparse0} reporting the $R$-th query $Q_R$ is ``above threshold'',
	    we have that $\forall \beta \in (0,1),
	    \Pr\inb{\exists i < R : \p(\T_i, Q_i) \ge \frac{(R+1)^{\frac{12\eps_1}{\eps_3}}}{\beta^{\frac{12\eps_1}{\eps_3}} + (R+1)^{\frac{12\eps_1}{\eps_3}}}  } \le \beta
      $.
      In other words, the algorithm does not stop too late.
    \item $\forall \beta \in (0,1)$, if for some $i$,
	    $
	\p(\T_i, Q_i) \ge 1-\frac{\beta^{\frac{12\eps_1}{\eps_3}}}{ 1+ \beta^{\frac{12\eps_1}{\eps_3}} } ,
	$
	then
	$\Pr\inb{a_i = \top | \forall j<i, a_j = \bot} \le \beta$.
	    In other words, the algorithm will likely halt on a query that is way above the threshold.
  \end{enumerate}
  \label{cor:gensparse0}
\end{corollary}

\subsection{Sparse vector for online private queries}
Next we show that one could replace the exact percentile oracles $p(\T_i,Q_i(\Data))$ with unbiased estimators $\pti$.
Assuming that we have unlimited access to the randomized queries $\set{Q_i(\Data)}$, we consider the following natural unbiased estimator for $p(\T_i,Q_i(\Data))$:
given iid samples $q_{i,1}, \cdots, q_{i,N}$, where for each $j$, $q_{i,j} \sim Q_i(\Data)$, we define $\pti \defeq \frac{1}{N} \sum_{j=1}^N \knuth{q_{i,j} \ge \T_i} $,
where $\knuth{q_{i,j} \ge \T_i}$ is the Iverson bracket defined by
\[\knuth{a \ge b} \defeq
\begin{cases}
	1, &\hbox{ if $a \ge b$,} \\
	0, &\hbox{ otherwise}.
\end{cases}
\]

Since $\pti$ is now a random function of the dataset, the usual Lipschitzness is not well-defined, unlike for the function $\p(\T_i,Q_i(\Data))$.
One approach of defining ``Lipschitzness'' for such a random function would be to consider the \emph{earth mover distance}.
This is what we will do next.

Let $\pti'$ be the analogous unbiased estimator for $p(\T_i,Q_i(\Data'))$ on a neighboring dataset $\Data'$.
By $\eps_1$-DP of $Q_i$, we have that
\[
	\E \pti = p(\T_i,Q_i(\Data)) \le e^{\eps_1} p(\T_i,Q_i(\Data')) = e^{\eps_1} \E \pti'.
\]
In order to adapt~\cref{alg:gensparse0}, ideally we would like a probabilistic version of  $\pti \le e^{\eps_1} \pti'$ to be true:
if there is a coupling between $\pti$ and $\pti'$ such that $\abs{\ln \pti - \ln \pti'} \le \eps_1$, then we can replace $p(\T_i,Q_i(\Data))$ with $\pti$ in~\cref{alg:gensparse0}.
This turns out to be too much to ask for in such a general setting.
We show in~\cref{lem:trivial-coupling} that a slightly weaker statement in indeed true.
This is the key lemma that leads us to~\cref{alg:gensparse1}.

\begin{lemma}
	Let $\set{X_1, \cdots, X_n}$ and $\set{Y_1, \cdots, Y_n}$ be two sequences of independent $\set{0,1}$ random variables, and let $X = \sum_{i=1}^n X_i$, $Y= \sum_{i=1}^n Y_i$.
	For any fixed $\eps_1 \in (0,1), \eps_0 \in (0,1)$, $\delta_0 \in (0,1)$,
let $C=2(e^{\eps_0+\eps_1} +1 + e^{\eps_0/2}) < 21$.

If $\E X \le e^{\eps_1} \E Y $, then under the trivial (independent) coupling between $X$ and $Y$,
	\[
		\Pr\inb{ X \ge e^{\eps_1+\eps_0} \cdot Y +  \frac{C}{\eps_0} \cdot \ln \frac{2}{\delta_0}} \le \delta_0.
	\]
	Equivalently, if we let $\Delta \defeq \frac{C\ln \frac{2}{\delta_0}}{\eps_0 \inp{e^{\eps_0+\eps_1} - 1}} = O\inp{\frac{1}{\eps_0^2} \ln\frac{1}{\delta_0}}$, then
	\[
	  \Pr\inb{ X+ \Delta  \ge e^{\eps_1+\eps_0} \cdot \inp{ Y +  \Delta}} \le \delta_0.
	\]
	\label{lem:trivial-coupling}
\end{lemma}
We defer the proof of this lemma to~\Cref{sec:proof-trivial-coupling}.
Now we are ready to describe the extended version of the AboveThreshold algorithm.
We now consider a potential function $\Pnd(x) = \frac{N x + \Delta }{N (1-x) + \Delta}$.
As an intuition, we will see that thanks to~\cref{lem:trivial-coupling}, if $Q_i$ is $\eps_1$-DP,
then for suitable choices of $\eps_0$ and $\Delta$, there exists a coupling in which, with high probability, $\ln \Pnd\inp{\pti}$ is $2(\eps_0+\eps_1)$-Lipschitz.

\begin{algorithm}
	\flushleft Input: $T, \delta, \eps_0, \eps_1, \eps_3, \beta, \pstar$, a stream of thresholds $\set{\T_i}$ and randomized queries $\set{Q_i(\Data)}$.

	Set $S = 2(\eps_1 + \eps_0)$, $\Delta =  \frac{C\ln\frac{8T}{\delta}}{\eps_0 \inp{e^{\eps_0+\eps_1} - 1}}$, and $N =  \frac{e^{\eps_0}\Delta}{\min\set{\pstar,1-\pstar}} \inp{\frac{T+1}{\beta}}^{\frac{6S}{\eps_3}} $

  Sample $\nu \sim \Lap\inp{\frac{2S}{\eps_3}}$

  For $i=1, \cdots, T$:
  \begin{itemize}
	  \item draw iid samples $\set{q_{i,1}, \cdots, q_{i,N}} \sim Q_i^N$
	  \item let $\pti \defeq \frac{1}{N} \sum_{j=1}^N \knuth{q_{i,j} \ge \T_i}$
    \item sample $\xi_i \sim \Lap\inp{\frac{4S}{\eps_3}}$
    \item if $\exp(\xi_i) \cdot \Pnd(\pti) \ge \exp\inp{\nu } \cdot \Pnd(\pstar)  $:
	    output $a_i = \top$ and halt
    \item otherwise output $a_i=\bot$
  \end{itemize}
  \caption{The ExtendedAboveThreshold algorithm.}
  \label{alg:gensparse1}
\end{algorithm}

\begin{theorem}
  For any fixed $\eps_0 \in (0,1)$, $\delta \in (0,1)$, $\beta \in (0,1)$ and an integer $T>1$,
let $S = 2(\eps_1 + \eps_0)$, $C=2(e^{\eps_0+\eps_1} +1 + e^{\eps_0/2}) < 21$, and $\Delta = \frac{C\ln\frac{8T}{\delta}}{\eps_0 \inp{e^{\eps_0+\eps_1} - 1}} = O\inp{\frac{1}{\eps_0^2} \ln\frac{T}{\delta}}$.
If for every $i$, $Q_i$ is $\eps_1$-DP,
then: 
\begin{enumerate}[(a)]
  \item \cref{alg:gensparse1} with the above parameters is $(\eps_3,\delta)$-DP.

    \item Conditional on \cref{alg:gensparse1} reporting the $R$-th query $Q_R$ is ``above threshold'',
	    we have that $\forall \beta \in (0,1),
	    \Pr\inb{\Pnd\inp{\p(\T_R, Q_R)} \le \inp{\frac{\beta}{R+1}}^{\frac{6S}{\eps_3}}  \cdot e^{-\eps_0} \cdot \Pnd(\pstar)} \le \beta + \delta/2
      $.
      In other words, the algorithm does not stop too early. Moreover,
      \[
	      \Pr\inb{\p(\T_R, Q_R) \le \frac{1}{2}e^{-\eps_0} \inp{\frac{\beta}{R+1}}^{\frac{6S}{\eps_3}}\pstar} \le \beta + \delta/2.
    \]
    \item Conditional on \cref{alg:gensparse1} reporting the $R$-th query $Q_R$ is ``above threshold'',
	    we have that $\forall \beta \in (0,1),
	    \Pr\inb{\exists i < R : \Pnd\inp{\p(\T_i, Q_i)} \ge \inp{\frac{R+1}{\beta}}^{\frac{6S}{\eps_3}}  \cdot e^{\eps_0} \cdot \Pnd(\pstar)} \le \beta + \delta/2
      $.
      In other words, the algorithm does not stop too late. Moreover,
      \[
	      \Pr\inb{\p(\T_R, Q_R) \ge 1 - \inp{1 + \frac{e^{\eps_0}}{2(1-\pstar)}\inp{\frac{R+1}{\beta}}^{\frac{6S}{\eps_3}} }^{-1} } \le \beta + \delta/2.
      \]
    \item $\forall \beta \in (0,1)$, if for some $i$,
	    $
	    \Pnd\inp{\p(\T_i, Q_i)} \ge e^{\eps_0}\inp{\frac{1}{\beta}}^{\frac{6S}{\eps_3}}  \Pnd(\pstar),
	$
	then
\[
	\Pr\inb{a_i = \top | \forall j<i, a_j = \bot} \le \beta+ \frac{\delta}{4T}.
\]
	In other words, on a query that is way above the threshold the algorithm will likely halt.
%
  \end{enumerate}
  \label{thm:gensparse1}
\end{theorem}

Before proving the theorem, we state the following sufficient condition for establishing $(\eps,\delta)$-DP.

\begin{lemma}
	Let $X$ and $Y$ be two random variables that share the same sample space and $\sigma$-algebra,
	If there exists constants $\delta>0, \eps>0$, and for any event $A$, there exists a joint event $\G \defeq \G(X,Y)$ on $X$ and $Y$ such that $\Pr[\G] \ge 1-\delta$,
	and
	\[
		e^{-\eps} \Pr[Y\in A | \G] \le \Pr[X\in A | \G] \le e^{\eps}\Pr[Y\in A | \G],
	\]
	then, $X$ and $Y$ also satisfies that 
	\[
		e^{-\eps}\inp{ \Pr[Y\in A ] - \delta} \le \Pr[X\in A] \le e^{\eps}\Pr[Y\in A] + \delta.
	\]
	\label{lem:couplingDP}
\end{lemma}
Informally, in order to show $(\eps,\delta)$-DP, it suffices to construct a coupling where, except with probability $\delta$, the two neighboring distributions satisfy $\eps$-DP.
It is worth noting that $X$ and $Y$ need not be independent. Thus one could optimize $\delta$ by constructing a coupling between $X$ and $Y$ that maximizes $\Pr[\G]$.
In addition, we note that the design of $\G$ and the coupling between $X$ and $Y$ can be dependent on the event $A$.
\begin{proof}
	For the first inequality,
	\begin{align*}
		\Pr[X \in A] \ge& \Pr[X\in A | \G] \cdot \Pr[\G] \\
		\ge& e^{-\eps}\Pr[Y \in A | \G] \cdot \Pr[\G] \\
		=& e^{-\eps}\inp{\Pr[Y \in A] - \Pr\inb{Y \in A, \overline{\G}}}\\
		\ge& e^{-\eps}\inp{\Pr[Y \in A] - \Pr\inb{\overline{\G}}} \\
		\ge& e^{-\eps}\inp{ \Pr[Y \in A ] - \delta}.
	\end{align*}

	For the second inequality,
	\begin{align*}
		\Pr[X \in A] =& \Pr[X \in A , \G]  + \Pr\inb{X \in A ,\overline{G}}\\
		\le& \Pr[X \in A | \G] \cdot \Pr[\G] + \Pr\inb{\overline{G}}\\
		\le& e^{\eps}\Pr[Y \in A | \G]\cdot  \Pr[\G] + \delta\\
		\le& e^{\eps} \Pr[Y \in A ] + \delta.
	\end{align*}
\end{proof}
%
Finally we prove~\cref{thm:gensparse1}.

{\noindent\emph{Proof of~\cref{thm:gensparse1}. }}
  For part (a), we follow the standard analysis of sparse vector.
  Fix any two neighboring datasets $\Data$ and $\Data'$.
  By~\cref{lem:couplingDP}, in order to show $(\eps_3,\delta)$-DP, it suffices to find a conditioning event $\G$, and a coupling between
  the output distribution of~\cref{alg:gensparse1} running on $\Data$ and $\Data'$, such that they are $\eps_3$-close except with probability $\delta$.
  Observe that in order to obtain the same output, it suffices if we can couple all the noisy tests of the form
  $e^{\xi_i}\cdot \Pnd\inp{\pti} \ge e^{\nu} \Pnd(\pstar)$.
  These tests depend only on two types of randomness: the perturbations to the current percentile (in the form of $\xi_i$),
  and the perturbations to the desired percentile (in the form of $\nu$).
  We denote these randomness by $\set{\xi_i}$ and $\nu$ when running on dataset $\Data$ ,
  and by $\set{\xi_i'}$ and $\nu'$ when running on $\Data'$.

  We consider the event that $a_R=\top$ and  $\forall i<R, a_i=\bot$.
  Let $\Phi* \defeq \Pnd(\pstar)$, and
  \begin{align*}
	  \Phi_i \defeq \Pnd\inp{\pti}\quad\quad&\hbox{ and }\quad\quad
	  \Phi_i' \defeq \Pnd\inp{\pti'}, \\
	  g \defeq \max_{i<R} \set{e^{\xi_i}\cdot \Phi_i} \quad\quad&\hbox{ and }\quad\quad
  g' \defeq \max_{i<R} \set{e^{\xi_i}\cdot \Phi_i'}.
  \end{align*}
  Now we are ready to specify the coupling.
  Given $\set{\xi_i}$ and $\nu$, we let
  $\nu' = \nu + \ln \frac{g'}{g}$, and $\xi_i' =
  \begin{cases}
    \xi_i, &\hbox{ if $i<R$}\\
    \xi_R + \ln \frac{g'}{g} + \ln \frac{\Phi_R}{\Phi_R'}, &\hbox{ if $i=R$}
  \end{cases}
  $.
  Then, it is not hard to check that under this coupling,
  \begin{align*}
	  g<e^{\nu} \Phi_* &\iff g' < e^{\nu'} \Phi_*, \\
	  e^{\xi_R} \Phi_R \ge e^{\nu} \Phi_* &\iff e^{\xi_R'} \Phi_R \ge e^{\nu'} \Phi_*.
  \end{align*}

  In the following we will abuse notation, and write $\Pr_{\Lap}[\xi_R]$ to denote the probability density function of the Laplace distribution.
  Then, let $a_i$ be the $i$-th output of the algorithm running on dataset $\Data$,
  and $a_i'$ be that of $\Data'$.
  \begin{align*}
	  \frac{\Pr[a_R=\top]}{\Pr[a_R'=\top]}
	  =& \frac{\Pr\inb{g<e^{\nu} \Phi_* \wedge  e^{\xi_R} \Phi_R \ge e^{\nu} \Phi_* }}{\Pr\inb{g'<e^{\nu'} \Phi_* \wedge  e^{\xi_R'} \Phi_R \ge e^{\nu'} \Phi_* }}\\
	  =&\frac{\Integrate{\Integrate{\Pr\inb{g<e^{\nu} \Phi_* \wedge  e^{\xi_R} \Phi_R \ge e^{\nu} \Phi_* \mid \xi_R, \nu} \cdot \Pr_{\Lap}[\xi_R] \cdot \Pr_{\Lap}[\nu]}{\nu,\R}}{\xi_R,\R}}{\Integrate{\Integrate{\Pr\inb{g'<e^{\nu'} \Phi_* \wedge  e^{\xi_R'} \Phi_R \ge e^{\nu'} \Phi_* \mid \xi_R', \nu' } \cdot \Pr_{\Lap}[\xi_R'] \cdot \Pr_{\Lap}[\nu']}{\nu',\R}}{\xi_R',\R}}  \\
	  \le& \sup_{\xi,\nu} \frac{\Pr_{\Lap}[\xi_R] \cdot \Pr_{\Lap}[\nu]}{\Pr_{\Lap}[\xi_R'] \cdot \Pr_{\Lap}[\nu']}, \quad\hbox{ by the coupling between $\xi_R,\xi_R'$ and $\nu,\nu'$}.
  \end{align*}

  Therefore, it remains to bound $\abs{\xi_R - \xi_R'}$ and $\abs{\nu - \nu'}$, which depends on the randomness involved in the probabilistic queries $\Phi_i$ and $\Phi_i'$.
  Thus we need to couple $\Phi_i$ and $\Phi_i'$.
  For the given $\eps_1,\eps_0$ (as specified in the theorem statement), we let $\delta_0 = \delta/T$, $X_1 = N\pti$,  and $Y_1=N\pti'$.
  Recall that $e^{-\eps_1} \E Y_1 \le \E X_1 \le e^{\eps_1} \E Y_1 $, then by~\cref{lem:trivial-coupling}, $X_1$ and $Y_1$ under the trivial coupling satisfies:
  \begin{align*}
	  \Pr\inb{ X_1 + \Delta \ge e^{\eps_1+\eps_0} \cdot \inp{Y_1 +  \Delta}} \le \frac{\delta}{4T},\\
	  \Pr\inb{ Y_1 + \Delta \ge e^{\eps_1+\eps_0} \cdot \inp{X_1 +  \Delta}} \le \frac{\delta}{4T}.
  \end{align*}
  Similarly if we let $X_2=N(1-\pti)$ and $Y_2 = N (1-\pti')$, then under the trivial coupling,
  \begin{align*}
	  \Pr\inb{ X_2 + \Delta \ge e^{\eps_1+\eps_0} \cdot \inp{Y_2 +  \Delta}} \le \frac{\delta}{4T},\\
	  \Pr\inb{ Y_2 + \Delta \ge e^{\eps_1+\eps_0} \cdot \inp{X_2 +  \Delta}} \le \frac{\delta}{4T}.
  \end{align*}
  We consider the following conditioning event:
  \begin{align*}
	  \G \defeq \set{\forall i \in [R],\quad \abs{\ln \inp{\frac{N \pti' +  \Delta}{N\pti + \Delta} }} \le  \eps_1+\eps_0 \quad\hbox{ and }\quad \abs{\ln \inp{\frac{N (1-\pti') +  \Delta}{N(1-\pti) + \Delta} }} \le  \eps_1+\eps_0}.
  \end{align*}
  By a union bound, we have $\Pr\inb{ \G } \ge 1 - \delta$.
  Conditional on $\G$, by triangle inequality we have:
  \begin{align*}
	  \forall i\in [R],\quad \abs{\ln \frac{\Phi_i}{\Phi_i'}} \le \abs{\ln \inp{\frac{N \pti' +  \Delta}{N\pti + \Delta} }} + \abs{\ln \inp{\frac{N (1-\pti') +  \Delta}{N(1-\pti) + \Delta} }} \le  2(\eps_1+\eps_0) = S.
  \end{align*}
  In other words, conditional on $\G$,
  \begin{align*}
	  \abs{\nu - \nu'} =&\abs{\ln \frac{g'}{g}} \le \max_{i<R} \abs{\ln \frac{\Phi_i}{\Phi_i'}} \le S \\
	  \abs{\xi_R - \xi_R'} \le&  \abs{\ln\frac{g'}{g}} + \abs{\ln \frac{\Phi_R}{\Phi_R'}}\le 2 \max_{i\le R} \abs{\ln \frac{\Phi_i}{\Phi_i'}} \le 2S.
  \end{align*}

  Now we are ready to bound
  \begin{align*}
	  \frac{\Pr[a_R=\top \mid \G ]}{\Pr[a_R'=\top \mid \G]}
	  \le  \frac{\Pr_{\Lap}[\xi_R] \cdot \Pr_{\Lap}[\nu]}{\Pr_{\Lap}[\xi_R \pm 2S] \cdot \Pr_{\Lap}[\nu \pm S]}
	  \le \exp(\eps_3).
  \end{align*}
    where the last inequality uses the probability density function of the two Laplace distributions.

    Finally consider the event that $R=T$ and $a_i=\bot$ for all $i \in [T]$, by a similar argument we have
  \begin{align*}
	  \frac{\Pr[a_T=\bot \mid \G ]}{\Pr[a_T'=\bot \mid \G]}
	  =& \frac{\Pr\inb{g<e^{\nu} \Phi_* \mid \G  }}{\Pr\inb{g'<e^{\nu'} \Phi_* \mid \G}}
	  \le  \frac{\Pr_{\Lap}[\xi_R] }{\Pr_{\Lap}[\xi_R \pm 2S] }
	  \le \exp(\eps_3).
  \end{align*}
    Since our choice of $R$ is arbitrary, this shows that conditioned on $\G$, we have $\eps_3$-DP for the output of our algorithm.
    Since $\Pr[\G] \ge 1-\delta$, by~\cref{lem:couplingDP} this concludes $(\eps_3,\delta)$-DP for the output unconditionally.

	\bigskip

	For part (b), we consider the events of non-concentration:
	\begin{align*}
	  \F_1 &\defeq \set{\nu : \abs{\nu} \ge \frac{2S}{\eps_3}\ln\frac{R+1}{\beta}}\\
	  \F_2 &\defeq \set{\xi_1, \cdots, \xi_{R}:  \exists i \in [R], \abs{\xi_i} \ge \frac{4S}{\eps_3}\ln\frac{R+1}{\beta}}\\
	  \F_3 &\defeq \set{\pt_1, \cdots, \pt_{R}:  \exists i \in [T],  \Pnd(\pti) > e^{\eps_0} \cdot \Pnd\inp{\E \pti}}.
	\end{align*}
	Then similar to part (a) we have
	\[
		\Pr[\F_1 \cup \F_2 \cup \F_3] \le \Pr[\F_1] + \Pr[\F_2] + \Pr[\F_3] \le \frac{\beta}{R+1} + R\frac{\beta}{R+1} + 2R\frac{\delta}{4T}\le \beta + \delta/2,
	\]
	where the bounds for $\F_1$ and $\F_2$ follows directly from CDF of the Laplace distribution, and the bound for $\F_3$ follows from a concentration bound (see~\cref{lem:concentrationX}).
	Therefore, conditional on avoiding $\F_1 \cup \F_2$,
	if the algorithm stops at the $k$-th iteration, we have that
	\begin{align*}
		\exp(\xi_k) \cdot \Pnd(\ptk) \ge \exp(\nu) \cdot \Pnd(\pstar)
		\implies  \Pnd(\ptk) \ge \inp{\frac{\beta}{R+1}}^{\frac{6S}{\eps_3}} \cdot \Pnd(\pstar).
	\end{align*}
	Next, conditioning further on avoiding $\F_3$, we have that
	\begin{align*}
		&\Pnd(\ptk) \le e^{\eps_0} \Pnd\inp{ \E \ptk}\\
		\implies &
		\Pnd\inp{ \E \ptk} \ge e^{-\eps_0} \cdot \Pnd(\ptk) \ge e^{-\eps_0}\inp{\frac{\beta}{R+1}}^{\frac{6S}{\eps_3}} \cdot \Pnd(\pstar).
	\end{align*}

	Let $N \ge  \frac{e^{\eps_0}\Delta}{\pstar}\inp{\frac{T+1}{\beta}}^{\frac{6S}{\eps_3}} $,
	then we have
	\begin{align*}
		&\Pnd(\pstar) = \frac{\pstar + \Delta/N}{1- \pstar + \Delta/N}\ge 2 \pstar, \\
		\implies&\Pnd\inp{ \E \ptk}\ge 2e^{-\eps_0} \inp{\frac{\beta}{R+1}}^{\frac{6S}{\eps_3}}\pstar, \\
		\implies&\p(\T_k, Q_k) = \E \ptk \ge \frac{e^{-\eps_0} \inp{\frac{\beta}{R+1}}^{\frac{6S}{\eps_3}}\pstar}{1 + e^{-\eps_0} \inp{\frac{\beta}{R+1}}^{\frac{6S}{\eps_3}}\pstar} \ge \frac{1}{2}e^{-\eps_0} \inp{\frac{\beta}{R+1}}^{\frac{6S}{\eps_3}}\pstar. 
	\end{align*}


	This concludes the proof.

	\bigskip

	For part (c), it will be similar to part (b),
	except that we consider
	\[
	  \F_3 \defeq \set{\pt_1, \cdots, \pt_{R}:  \exists i \in [R],  \Pnd(\pti) < e^{-\eps_0} \cdot \Pnd\inp{\E \pti}}.
  \]
  As before we still have
	\[
		\Pr[\F_1 \cup \F_2 \cup \F_3] \le \Pr[\F_1] + \Pr[\F_2] + \Pr[\F_3] \le \frac{\beta}{R+1} + R\frac{\beta}{R+1} + 2R\frac{\delta}{4R}\le \beta + \delta/2.
	\]

	Then, conditioning on avoiding $\F_1,\F_2,\F_3$, we have that
	\begin{align*}
		\Pnd\inp{ \E \ptk} \le e^{\eps_0} \cdot \Pnd(\ptk) \le e^{\eps_0}\inp{\frac{R+1}{\beta}}^{\frac{6S}{\eps_3}} \cdot \Pnd(\pstar).
	\end{align*}

	Let $N \ge \frac{e^{\eps_0}\Delta}{1-\pstar}\inp{\frac{T+1}{\beta}}^{\frac{6S}{\eps_3}} $,
	then we have
	\begin{align*}
		&\Pnd(\pstar) = \frac{\pstar + \Delta/N}{1- \pstar + \Delta/N}\le \frac{1}{2(1-\pstar)}, \\
		\implies&\Pnd\inp{ \E \ptk}\le \frac{e^{\eps_0}}{2(1-\pstar)}\inp{\frac{R+1}{\beta}}^{\frac{6S}{\eps_3}} ,\\
		\implies&\p(\T_k, Q_k) = \E \ptk \le 1 - \inp{1 + \frac{e^{\eps_0}}{2(1-\pstar)}\inp{\frac{R+1}{\beta}}^{\frac{6S}{\eps_3}} }^{-1}.
	\end{align*}

	\bigskip


	For part (d), 
	observe that the test $\exp(\xi_i) \cdot \Pnd(\pti) \ge \exp\inp{\nu } \cdot \Pnd(\pstar)  $ will pass if
	$\xi_i \ge -\frac{4S}{\eps_3} \ln \frac{1}{\beta}$,
	$\nu \le \frac{2S}{\eps_3} \ln \frac{1}{\beta}$,
	and
	$\Pnd(\pti) \ge e^{-\eps_0/2} \Pnd(\E\pti)$.
	Similar to part (b) and (c), we get that this will happen except with probability $\beta + \frac{\delta}{4T}$.

	\qed

%% file: sparse_vector.tex
\subsection{A more efficient sparse vector for a one-sided guarantee}
In this subsection we consider searching for the unknown ``good'' threshold $\T$ for~\cref{alg:thresholding} in a more efficient yet private way.
The idea is that, instead of trying to tackle adversarily chosen randomized online queries, here we design better queries for our algorithm.

Specifically, let
$Q(\Data)$ be a distribution dependent on dataset $\Data$, and let $q\sim Q$.
Let $p(\T, Q) \defeq \Pr_{q\sim Q} [q \ge \T]$.
Then, given $\pstar \in (0,1)$, our goal is to find $\Tstar \defeq \max\set{\T: p(\T, Q) \ge \pstar}$ in a differentially private way.

Since $\Tstar$ can be very sensitive for neighboring datasets (see~\Cref{sec:percentileDP} for an example for $\pstar=1/2$),
outputting $\Tstar$ directly would not be private.
The relaxed goal is to find, with high probability, a private threshold $\Tt$ so that:

\begin{center}
$\Tt$ is almost as large as $\Tstar$, and $p(\Tt,Q)$ is not much smaller than $\pstar$.
\end{center}

It is worth noting that, due to the one-sided nature of our goal (instead of asking $p(\Tt,Q)$ to be close to $\pstar$, we only want $p(\Tt,Q)$ to be not much smaller than $\pstar$),
we find it much more convenient to shift the target by a constant factor: from $\pstar$ to a smaller target $\approx \beta^{\frac{6\eps_1}{\eps_3}} \cdot \pstar$.
Such a tradeoff enables us to find a $\Tt$ that is closer to $\Tstar$, at the cost of a potentially smaller $p(\Tt,Q)$.
In the settings that we consider, a higher $\Tt$ allows for better ``quality'' of the selected candidate, while a larger $\pstar$ is usually only for smaller computational cost.
\begin{algorithm}
	\flushleft Input: $R, \delta, \eps_0, \eps_1, \eps_3, \beta, \pstar$, and sampling access to $Q(\Data)$.

	Set $S = \eps_1 + \eps_0$, $\Delta =  \frac{C\ln\frac{4R}{\delta}}{\eps_0 \inp{e^{\eps_0+\eps_1} - 1}}$, $N =  \frac{3\Delta e^{\eps_0/2}}{\pstar} \cdot \beta^{\frac{- 12 S}{\eps_3}} (R+1)^{\frac{6S}{\eps_3}} $, and $\Lambda = \frac{\eps_0}{2} + \frac{6S}{\eps_3} \ln \frac{1}{\beta}$

  Sample $\nu \sim \Lap\inp{\frac{2S}{\eps_3}}$, and draw iid samples $\set{q_1, \cdots, q_N} \sim Q^N$ 

  For $i=1, \cdots, R $:
  \begin{itemize}
    \item let $\T_i = 1 - \frac{i-1}{R-1}$
    \item let $\pti \defeq \frac{1}{N} \sum_{j=1}^N \knuth{q_i \ge \T_i}$
    \item sample $\xi_i \sim \Lap\inp{\frac{4S}{\eps_3}}$
    \item if $\exp(\xi_i) \cdot (N \pti + \Delta) \ge \exp\inp{\nu - \Lambda} \cdot (N \pstar + \Delta)  $:
      \begin{itemize}
	\item output $\T_i$ and 
	  halt
      \end{itemize}
  \end{itemize}
  \caption{The FindPercentileThreshold algorithm.}
  \label{alg:sparse1}
\end{algorithm}

\begin{theorem}
  Let $Q$ be a $\eps_1$-DP distribution.
  For any fixed $\eps_0 \in (0,1)$, $\delta \in (0,1)$, $\beta \in (0,1)$ and an integer $R>1$,
let $S = \eps_1 + \eps_0$, $C=2(e^{\eps_0+\eps_1} +1 + e^{\eps_0/2}) < 21$, and $\Delta = \frac{C\ln\frac{4R}{\delta}}{\eps_0 \inp{e^{\eps_0+\eps_1} - 1}} = O\inp{\frac{1}{\eps_0^2} \ln\frac{R}{\delta}}$.
Then the following holds for the output $\Tt$ of~\cref{alg:sparse1} with the above parameters:
\begin{enumerate}[(a)]
  \item $\Tt$ is $(\eps_3,\delta)$-DP.
    \item
the algorithm does not stop too early:
      \[
	      \Pr\inb{p(\Tt, Q) \le \frac{e^{-\eps_0/2}}{3}\inp{\frac{\beta^2 }{R+1}}^{\frac{6S}{\eps_3}}  \cdot \pstar} \le \beta + \delta/2.
    \]

    \item $Pr\inb{\Tt \le \Tstar - \frac{1}{R}} \le 2\beta + \frac{\delta}{2R}$.
      In other words, the algorithm does not stop too late.
  \end{enumerate}
  \label{thm:sparse1}
\end{theorem}

\begin{proof} (Sketch)
	For part (a), this basically follows from the same proof of~\cref{thm:gensparse1} part (a), except the following changes:
	\begin{itemize}
		\item We consider $\Pnd(x) = N x + \Delta$. It is not hard to see that the proof only relies on the fact that $\Pnd(x)$ is monotone, and $\Pnd\inp{\pti}$ can be coupled with $\Pnd\inp{\pti'}$ multiplicatively.
		\item Here we can re-use randomness, due to the fact that we essentially have the same distribution, and only need to change $\T_i$. It is worth noting that we did not require independence of the $\set{\pti}$ since we only used union bound.
		\item We have also shifted the target of $\Pnd(\pstar)$ multiplicatively. However it does not affect privacy, since one can view such a shift as considering a different $\pstar$ to begin with.
	\end{itemize}
	\bigskip

	For part (b), similarly we consider the events of non-concentration:
	\begin{align*}
	  \F_1 &\defeq \set{\nu : \abs{\nu} \ge \frac{2S}{\eps_3}\ln\frac{R+1}{\beta}}\\
	  \F_2 &\defeq \set{\xi_1, \cdots, \xi_{R}:  \exists i \in [R], \abs{\xi_i} \ge \frac{4S}{\eps_3}\ln\frac{R+1}{\beta}}\\
	  \F_3 &\defeq \set{\pt_1, \cdots, \pt_{R}:  \exists i \in [R],  N \pti + \Delta > e^{\eps_0/2} \inp{ N p(\T_i,Q) + \Delta }}.
	\end{align*}
	Then we have
	\[
		\Pr[\F_1 \cup \F_2 \cup \F_3] \le \Pr[\F_1] + \Pr[\F_2] + \Pr[\F_3] \le \frac{\beta}{R+1} + R\frac{\beta}{R+1} + R\frac{\delta}{2R}\le \beta + \delta/2,
	\]
	where the bounds for $\F_1$ and $\F_2$ follows directly from CDF of the Laplace distribution, and the bound for $\F_3$ follows from a concentration bound (see~\cref{lem:concentrationX}).
	Therefore, conditional on avoiding $\F_1 \cup \F_2$,
	if the algorithm stops at the $k$-th iteration, we have that
	\begin{align*}
	  &\exp(\xi_k) \cdot (N \ptk + \Delta) \ge \exp(\nu - \Lambda) \cdot (N \pstar + \Delta)\\
	  \implies &\inp{\frac{R+1}{\beta}}^{\frac{4S}{\eps_3}} \cdot (N \ptk + \Delta) \ge \inp{\frac{\beta}{R+1}}^{\frac{2S}{\eps_3}} \cdot \beta^{\frac{6S}{\eps_3}} \cdot e^{-\eps_0/2} \cdot (N \pstar + \Delta) \\
	  \implies & N\ptk \ge \inp{\frac{1}{R+1}}^{\frac{6S}{\eps_3}} \beta^{\frac{12S}{\eps_3}}\cdot N \pstar - \Delta.
	\end{align*}
	Set $N =  \frac{3\Delta e^{\eps_0/2}}{\pstar} \cdot \beta^{\frac{- 12 S}{\eps_3}} (R+1)^{\frac{6S}{\eps_3}} $,
	then we have
	\begin{align}
	N\ptk \ge 2\Delta e^{\eps_0/2} = \frac{2}{3}\inp{\frac{1}{R+1}}^{\frac{6S}{\eps_3}} \beta^{\frac{12S}{\eps_3}} \cdot N \pstar .
		\label{eq:nptk}
	\end{align}

	Next, conditioning further on avoiding $\F_3$, we have that
	\begin{align*}
		&N \ptk + \Delta \le e^{\eps_0/2} \inp{ N p(\Tt,Q) + \Delta }\\
		\implies &
		N p(\Tt,Q) \ge e^{-\eps_0/2} N \ptk - \Delta \ge \Delta= \frac{e^{-\eps_0/2}}{3}\inp{\frac{1}{R+1}}^{\frac{6S}{\eps_3}} \beta^{\frac{12S}{\eps_3}} \cdot N \pstar \\
		\implies &
		p(\Tt,Q) \ge \frac{e^{-\eps_0/2}}{3}\inp{\frac{1}{R+1}}^{\frac{6S}{\eps_3}} \beta^{\frac{12S}{\eps_3}} \cdot \pstar.
	\end{align*}

	This concludes the proof.

	\bigskip


	For part (c), as soon as $\T_i \le \Tstar$, we have $p(\T_i, Q) \ge \pstar$.
	Therefore the test $\exp(\xi_i) \cdot (N \pti + \Delta) \ge \exp\inp{\nu - \Lambda} \cdot (N \pstar + \Delta)  $ will pass if
	$\xi_i \ge -\frac{4S}{\eps_3} \ln \frac{1}{\beta}$,
	$\nu \le \frac{2S}{\eps_3} \ln \frac{1}{\beta}$,
	and
	$N \pti + \Delta \ge e^{-\eps_0/2} \inp{ N p(\T_i,Q) + \Delta }$.
	Similar to part (b), we get that this will happen except with probability $2\beta + \frac{\delta}{2R}$.
	In other words, the probability of not halting after the first iteration with $\T_i \le \Tstar$ is at most $2\beta + \frac{\delta}{2R}$.
\end{proof}

Finally, by combining~\cref{thm:sparse1} and~\cref{thm:thresholdingDP}, we get the following:
\begin{theorem}
	\label{thm:main}
Fix any $\eps_1>0, \eps_0 \in [0,1], \beta > 0, R \in \mathbb{N}$.
Suppose that there are $\eps_1$-DP algorithms $\M_1, \ldots, \M_K$ and let $\tau^*(D) = \max_i \Median(\M_i(D))$. There is an algorithm $\M$ that on any dataset either outputs $\perp$, or selects an $i$ and a sample $x$ from $\M_i(D)$ such that
(a) $\M$ is $(2\eps_1+\eps_0, \delta)$-DP,
(b) Except with probability $\beta+\delta/R$, $x$ has quality at least  $\tau^* - \frac{1}{R}$,
(c) The number of calls $\widetilde{T}$ that the algorithm makes to any $\M_i(D)$ satisfies
\begin{align*}
	\E \widetilde{T} \le O\inp{K\inp{\frac{R+1}{\beta^2}}^{6+\frac{12\eps_1}{\eps_0}} \inp{\frac{\ln\frac{R}{\delta}}{\eps_0^2} + \frac{1}{\beta}} }
	\;\;;\;\;
	\widetilde{T} \le O\inp{K\inp{\frac{R+1}{\beta^2}}^{6+\frac{12\eps_1}{\eps_0}} \inp{\frac{\ln\frac{R}{\delta}}{\eps_0^2} + \frac{\ln \frac{1}{\eps_0}}{\beta}} }.
\end{align*}
  Furthermore,
	$\Pr\inb{ \M \hbox{outputs $\bot$} } \le \beta+\delta$.
\end{theorem}

Here we set $p_1 =\frac{1}{12K}\cdot \inp{\frac{\beta^2}{R+1}}^{6+\frac{12\eps_1}{\eps_0}}$, and $\gamma = p_1 \beta$, $\eps_3=\eps_0$.

%% file: applications.tex
\section{Applications}
\label{sec:applications}
\subsection{Hyperparameter selection}
Suppose that we are given $K$ choices of hyperparameters,
and for each choice $i \in [K]$, there is a differentially private learning algorithm $\M_i$.
Given a training dataset $\Data_1$,
$\M_i(\Data_1)$ is a randomized mechanism that returns a model,
which we often denote as $m$.
Next, for a validation dataset $\Data_2$, we let $\widetilde{q_i}(m,\Data_2)$ be the validation score of model $m$  and hyperparameter $i$.
Then the goal of hyperparameter selection is to find a pair $(m,i_*)$, that approximately maximizes the validation score.

It is worth noting that the dependencies on the validation set are only through the scoring functions $\widetilde{q_i}$, which are usually \emph{counting queries} and thus have small sensitivity.
This is the setting we will consider. Therefore, we let $q_i \defeq \widetilde{q_i} + \Lap\inp{\frac{1}{n \eps_2} }$, where $n$ is the size of the validation set.
    Then, we define $Q_i(\Data_1,\Data_2)$ to be the distribution of $q_i(m, \Data_2)$ when $m\sim \M_i(\Data_1)$.
    Finally we let $Q$ be the distribution of $Q_i$ when we draw $i$ uniformly from $[K]$.


    Then, in order to apply~\cref{thm:main} or~\cref{thm:maxRandDP}, 
  it remains to verify that $Q$ is differentially private with respect to both datasets.
  \begin{lemma}
    The distribution $Q(\Data_1,\Data_2)$ defined as above is always $\eps_2$-DP for the validation set $\Data_2$. Moreover:

    if $\set{\M_i}_{i=1}^K$ are $\eps_1$-DP learning algorithms, then $Q(\Data_1,\Data_2)$ is $\eps_1$-DP for the training set $\Data_1$;

    if $\set{\M_i}_{i=1}^K$ are $(\eps_1,\delta_1)$-DP learning algorithms, then $Q(\Data_1,\Data_2)$ is $(\eps_1,\delta_1)$-DP for $\Data_1$.
  \end{lemma}
  \begin{proof}
    First for $\Data_2$, notice that for neighboring $\Data_2$ and $\Data_2'$, $\widetilde{q_i}$ changes by at most $1/n$, thus $\forall m \in \supp\set{\M_i}$ and $\forall t, \xi\in \R$, there exists $\nu: \abs{\nu - \xi} \le 1/n$ such that
    \begin{align*}
      \Pr\inb{q_i(m,\Data_2') = t} &=\Pr\inb{\widetilde{q_i}(m,\Data_2') + \xi = t} = \Pr[\xi = t - \widetilde{q_i}(m,\Data_2')], \\
      \Pr\inb{q_i(m,\Data_2) = t} &=\Pr\inb{\widetilde{q_i}(m,\Data_2) + \nu = t}= \Pr[\nu = t - \widetilde{q_i}(m,\Data_2)].
    \end{align*}
    It is worth noting that this holds for every $m$ in the support.
    Then $\eps_2$-DP for $\Data_2$ follows from the fact that $\xi$ and $\nu$ follow the same $\Lap\inp{\frac{1}{n\eps_2}}$ distribution and $\abs{\xi-\nu}\le 1/n$.

    Then for $\Data_1$, note that the dependency of $Q_i$ on $\Data_1$ is only through $\M_i(\Data)$, which is $\eps_1$-DP. Thus for every $i$, $Q_i(\Data_1,\Data_2)$ is $\eps_1$-DP for $\Data_1$, thus $Q$ is also $\eps_1$-DP for $\Data_1$.

    Similarly if $\M_i(\Data)$ is $(\eps_1,\delta_1)$-DP,
    we have that for every $i$, $Q_i(\Data_1,\Data_2)$ is $(\eps_1,\delta_1)$-DP for $\Data_1$, thus $Q$ is also $(\eps_1,\delta_1)$-DP for $\Data_1$.
  \end{proof}

\subsection{Adaptive Data Analsis Beyond Low Sensitivity Queries}

Our results immediately have applications to designing differentially private algorithms where interemediate steps select the best amongst various private options. Since DP allows us to prove generalization bounds, these results have implications for adaptive data analysis too.

As an example, consider a data analysis algorithm which as an intermediate step runs $k$-means clustering (or rank-$k$ PCA). Often in practice, one tries several values of $k$ and picks the best one according to some criteria (see e.g. Garg and Kalai~\cite{GargK18}). While there are differentially private variants of the base problem of $k$-means, naively selecting the best would require us to account for the privacy cost of computing all the $k$-means objectives, for different value of $k$. Theorem~\ref{thm:main} allows us to select the best of these without any asymptotic overhead in privacy cost.

\subsection{Generalizations of the Exponential Mechanism}

The exponential mechanism solves the selection problem when the score functions are Lipschitz. Several variants of the Exponential Mechanism have been proposed in previous work. We next show that several of these can be derived as corollaries of our main result, by defining appropriate private variants of the score function.
\begin{theorem}
	Let $\{q_i(\cdot)\}_{i=1}^K$ be a set of score functions mapping datasets to reals. Let $i^\star(D) = \argmax_{i} q_i(D)$ and $q^\star(D) = \max_i q_{i}(D)$.
	\begin{description}
		\item[Exponential Mechanism] Suppose that each $q_i$ has sensitivity at most $s$. Then there is an $\eps$-DP mechanism that outputs an $i$ such that $q_i(D) \geq q^\star(D) - O(s\log \frac K \beta / \eps)$ except with probability $\beta$.
		\item[Generalized Exponential Mechanism~\cite{RaskhodnikovaS16}] Suppose that $q_i$ has sensitivity at most $s_i$. Then there is an $(\eps,\delta)$-DP mechanism that outputs an $i$ such that $q_i(D) \geq q^\star(D) - O(s_{i^\star}\log \frac K \beta / \eps)$ except with probability $\beta$.
		\item[Margin-based Mechanism $\mathcal{A}_{dist}$~\cite{SmithT13, BeimelNS13}] Suppose that each $q_i$ has sensitivity at most $s$. There is an $\ed$-DP mechanism $\M$ that outputs $i^\star$ except with probability $\beta$ whenever $q^\star \geq q^{i} + \Omega(s\log \frac 1 {\beta\delta} / \eps)$ for all $i\neq i^\star$.
		\item[Generalized Smooth Sensitivity Exponential Mechanism] Suppose that $q_i$ has $\left(\eps/(4\ln \frac 2 \delta)\right)$-smoothed sensitivity at most $s_i$. Then there is an $\ed$-DP mechanism that outputs an $i$ such that $q_i(D) \geq q^\star(D) - O(s_{i^\star}\log \frac K \beta / \eps)$ except with probability $\beta$.

	\end{description}
\end{theorem}
\begin{proof}
	For the first part, let $\M_i(D) = \left(i, q_i(D) + Lap(\frac{s}{\eps})\right)$. Then applying~\cref{thm:maxRandDP}, we get an outcome with score at least $q^\star(D)$ except with probability $\beta$ (by setting $\gamma = \beta/K$). Since the number of runs of any $q_i$ is at most $\tilde{O}(K/\beta)$ (except with probability $\beta$), the largest of the Laplacian r.v.'s is bounded by $O(\frac{s \log (K/\beta)}{\eps})$. This implies that $q_i(D)$ where $i$ is the option chosen by the algorithm is at lest $q^\star(D) - O(\frac{s \log (K/\beta)}{\eps})$.

  The second part is similar, except that we set $\M_i(D) = \left(i, q_i(D) -\frac{2s_i\log K/\beta}{\eps} + Lap(\frac{s_i}{\eps})\right)$. This shift ensures the realized score is no larger than $q_i(D)$ for all calls to $\M_i(D)$. Now the median of $\M_{i^\star}$ is at least $q^\star(D) - -\frac{2s_{i^\star}\log K/\beta}{\eps}$, which implies the claim.

  For the third part, consider the truncated Laplace distribution $\mathrm{TLap}^{(T)}(\lambda)$ that samples from the Laplace distribution with parameter $\lambda$, conditioned on the output being in $[-T\lambda, T\lambda]$. It can be checked~\cite{trunclap} that the mechanism $\M_i(D) = q_i(D) + \mathrm{TLap}^{(\log \frac 1 \delta)}(\frac{s}{\eps})$ satisfies $\ed$-DP when $\eps<1, \delta<1/4$. The claim follows by applying~\cref{thm:maxRandDPapprox}.

	The fourth part is similar to the Generalized exponential mechanism, except that we add noise from smooth-sensitivity-scaled Laplacian distribution using the Smoothed Sensitivity framework of~\cite[Cor. 2.4]{NissimRS07}. As long as $\eta > \eps / \log (K/\beta)$ (which is ensured when we set $\eta = \eps / 4\ln \frac 2 \delta)$ with $\delta < \beta/K$), it can be verified that the $2s_i$ is a smooth upper bound on the sensitivity of $q_i(D) -\frac{2s_{i^\star}\log K/\beta}{\eps}$. The claim follows by a simple computation.
\end{proof}

\subsection{Private Amplification}
Gupta et al.~\cite{GLMRT} study the question of private amplification: given a DP algorithm that gets a certain utility in expectation, can we convert it into one that gets close to that utility with high probabilty? Their motivation came from combinatorial optimization problems, where they showed appoximation algorithms with certain guarantees in expectation. Using Markov's inequality, one can convert the expectation guarantee to one that ensures a utility bound with some probability $p$. Applying our results, one gets an algorithm that ensures that utility with high probability. This improves on the private amplification theorem proven in~\cite{GLMRT}.

%% file: deferred.tex
\section{Deferred Proofs}

\subsection{Proof of~\cref{thm:thresholdingDP}}
\label{sec:proof-thresholding}
We restate~\cref{thm:thresholdingDP} here for convenience.
\begin{theorem}
  Fix any $\eps_1>0, \eps_0 \in [0,1], \delta_1>0, \gamma \in [0,1]$.
  Let $T$ be any integer such that $T\ge \max\set{\frac{1}{\gamma}\ln \frac{2}{\eps_0 } , 1+\frac{1}{e\gamma} }$, 
  Then~\cref{alg:thresholding} with these parameters satisfies the following:
  \begin{enumerate}[(a)]
	  \item Let $\Aout(\Data)$ be the output of~\cref{alg:thresholding}, then 
		  \[\Pr[\Aout(\Data) = (x,q)] \propto \Pr_{(\xt,\qt) \sim Q(\Data)}[(\xt,\qt) = (x,q)].\]
	  \item If $Q$ is $\eps_1$-DP, then the output is $(2\eps_1 + \eps_0)$-DP.
	  \item If $Q$ is $(\eps_1,\delta_1)$-DP, then the output is $\inp{2\eps_1 + \eps_0, \; 3 e^{2\eps_1 + \eps_0}\cdot\frac{\delta_1}{\gamma}}$-DP.
    \item Let $\widetilde{T}$ be the number of iterations of the algorithm, and let $p_1 = \Pr_{q \sim Q(\Data)} \inb{q \ge \T}$, then 
	    \[\E \widetilde{T} \le \frac{1}{p_1(1-\gamma) + \gamma} \le \min\set{\frac{1}{p_1}, \frac{1}{\gamma}}.\]
    \item
  Furthermore, 
  $\Pr\inb{ \hbox{output $\bot$} } \le \frac{(1-p_1) (1+ \eps_0/2)}{p_1}\gamma$.
  \end{enumerate}
\end{theorem}

\begin{proof}
	
	For part (a),
  let $ p(x,q) \defeq \Pr_{(\xt,\qt) \sim Q(\Data)}[(\xt,\qt) = (x,q)]$.
  Given a threshold $\T$, we let $p_1 = \Pr_{q \sim Q(\Data)} \inb{q \ge \T}$, and $p_1' = \Pr_{q \sim Q(\Data')} \inb{q \ge \T}$.
  Then we have
  \begin{align*}
	  \Pr\inb{ \Aout(\Data) = (x,q) } =& \sum_{j=1}^{T}\Pr\inb{ \Aout(\Data) = (x,q) \wedge \hbox{ stops after $j$ steps} } \\
	  =&\sum_{j=1}^T \inp{(1-p_1)(1-\gamma)}^{j-1} \cdot p(x,q) \\
	  =& p(x,q) \cdot \frac{1 -\inp{(1-p_1)(1-\gamma)}^T }{1-(1-p_1)(1-\gamma)}. 
  \end{align*}
  Note that $p_1$ only depends on $\T$ and not on $(x,q)$, and $\gamma$ is a constant, therefore we have $\Pr\inb{ \Aout(\Data) = (x,q) } \propto p(x,q)$.

  For part (b),
  since $Q$ is $\eps_1$-DP, we have that $p_1$ is $\eps_1$-close to $p_1'$, and $1-p_1$ is also $\eps_1$-close to $1-p_1'$.
  Let $ p'(x,q) \defeq \Pr_{(\xt,\qt) \sim Q(\Data')}[(\xt,\qt) = (x,q)]$, then we also have $p(x,q)$ is $\eps_1$-close to $p'(x,q)$.
  \begin{align*}
	  \frac{\Pr\inb{ \Aout(\Data) = (x,q) } }{\Pr\inb{ \Aout(\Data') = (x,q) }} =&\frac{p(x,q)}{p'(x,q)} \cdot  \frac{1 -\inp{(1-p_1)(1-\gamma)}^T }{1-\inp{(1-p_1')(1-\gamma)}^T} \cdot \frac{1 -(1-p_1')(1-\gamma) }{1-(1-p_1)(1-\gamma)} \\
    \le & \exp(\eps_1) \cdot \frac{1}{1-(1-\gamma)^T} \cdot \frac{p_1'(1-\gamma) + \gamma}{p_1(1-\gamma) + \gamma} \\
    \le & \exp(\eps_1) \cdot \frac{1}{1- \eps_0/2} \cdot \exp(\eps_1), \quad \hbox{ if $T \ge \frac{1}{\gamma}\ln \frac{2}{\eps_0} $} \\
    \le & \exp(2\eps_1 + \eps_0), \quad \hbox{ if $\eps_0 \le 1$}.
  \end{align*}

    Next we consider the event of outputting $\bot$ on dataset $\Data$.
  \begin{align*}
	  \Pr\inb{ \Aout(\Data) = \bot } 
    =& \inp{\sum_{j=1}^{T}\Pr\inb{ \Aout(\Data) = \bot \wedge \hbox{ stops after $j$ steps} }} +  \Pr\inb{ \hbox{not stopping after $T$ steps}}\\
    =&\inp{\sum_{j=1}^T \inp{(1-p_1)(1-\gamma)}^{j-1} \cdot (1-p_1) \gamma} + \inp{(1-p_1)(1-\gamma)}^T  \\
    =& (1-p_1)\gamma \cdot \frac{1 -\inp{(1-p_1)(1-\gamma)}^T }{1-(1-p_1)(1-\gamma)}+ \inp{(1-p_1)(1-\gamma)}^T  \\
    =&  \frac{(1-p_1)\gamma -(1-p_1)^{T+1}(1-\gamma)^T\gamma +  (1-p_1)^T(1-\gamma)^T - (1-p_1)^{T+1}(1-\gamma)^{T+1}}{1-(1-p_1)(1-\gamma)} \\
    =&  \frac{(1-p_1)\gamma  +  p_1(1-p_1)^T(1-\gamma)^T }{1-(1-p_1)(1-\gamma)}\\
    =&  (1-p_1)\gamma \cdot \frac{1 +  \frac{p_1}{\gamma}(1-p_1)^{T-1}(1-\gamma)^T }{p_1(1-\gamma) + \gamma}.
  \end{align*}
  Similarly, we have,

  \begin{align*}
  \frac{\Pr\inb{ \Aout(\Data) = \bot }  }{\Pr\inb{ \Aout(\Data') = \bot } } =&\frac{1-p_1}{1-p_1'} \cdot  \frac{1 +\frac{p_1}{\gamma} (1-p_1)^{T-1}(1-\gamma)^T }{1+ \frac{p_1'}{\gamma}(1-p_1')^{T-1}(1-\gamma)^T} \cdot \frac{p_1'(1-\gamma) + \gamma}{p_1(1-\gamma) + \gamma} \\
  \le& \exp(\eps_1) \cdot \inp{ 1 + p_1(1-p_1)^{T-1} \cdot \frac{1}{\gamma}(1-\gamma)^T } \cdot \exp(\eps_1) \\
  \overset{(\dagger)}{\le}& \exp(\eps_1) \cdot \inp{ 1 + \frac{1}{e (T-1)  \gamma} \cdot (1-\gamma)^T} \cdot \exp(\eps_1), \quad \hbox{ by AM-GM inequality} \\
  \le& \exp(2\eps_1) \cdot \inp{1+ \eps_0/2}, \quad \hbox{ if $T \ge \max\set{\frac{1}{\gamma}\ln \frac{2}{\eps_0 } , 1+\frac{1}{e\gamma} }$} \\
  \le&\exp(2\eps_1 + \eps_0),
  \end{align*}
  where $(\dagger)$ follows from AM-GM inequality: recall that $T>1$ is an integer, and $0 \le p_1\le 1$, then
  \[
    (T-1) p_1 (1-p_1)^{T-1} \le \inp{\frac{(T-1)p_1 + (T-1) (1-p_1)}{ T}}^T  = \inp{1 - \frac{1}{T}}^T \le e^{-1}.
  \]

  \bigskip

  This concludes part (b). For part (c), it is worth noting that the privacy does not degrade as we increase $T$ (the number of iterations).

  We consider any event $E$ on the output of~\cref{alg:thresholding}. Note that $E$ will be a set of tuples $(x,q)$, and possibly contain $\bot$.
  Let $\Aout(\Data)$ be the output of~\cref{alg:thresholding} on dataset $\Data$, and $\Aout(\Data')$ be the output on a neighboring dataset $\Data'$.
  If $\bot \in E$, then clearly $\Pr\inb{\Aout(\Data) \in E} = \Pr\inb{\Aout(\Data) \in E \setminus\set{\bot}} + \Pr\inb{\Aout(\Data) = \bot}$.
  In the following we will bound the two terms separately.
  For the first term, we consider any event $F$ that does not contain $\bot$.
  Let 
  \begin{align*}
	  p \defeq \Pr_{\inp{x,q} \sim Q(\Data)} \inb{(x,q) \in F }  \quad \quad \hbox{ and } &\quad \quad p' \defeq \Pr_{\inp{i,m,q} \sim Q(\Data')} \inb{(x,q) \in F }, \\
  p_1 \defeq \Pr_{q \sim Q(\Data)} \inb{q \ge \T}  \quad \quad \hbox{ and } & \quad \quad p_1' \defeq \Pr_{q \sim Q(\Data')} \inb{q \ge \T}.
  \end{align*}
  If $Q$ is $(\eps_1,\delta_1)$-DP, then we know that $p \le e^{\eps_1} p' + \delta_1$, and $p_1' \le e^{\eps_1} p_1 + \delta_1$, or equivalently that $p_1 \ge \max\set{0,  p_1' - \delta}e^{-\eps_1}$.
  Also notice that $p \le p_1$ and $p' \le p_1'$.
  Then, by calculations in part (a), we have the following upperbound:
  \begin{align*}
	  \Pr\inb{\Aout(\Data) \in F } =& p \cdot \frac{1 -\inp{(1-p_1)(1-\gamma)}^T }{1-(1-p_1)(1-\gamma)} \\
    \le & \frac{p }{p_1(1-\gamma) + \gamma} \\
    \le &  \frac{e^{\eps_1} \cdot  p' + \delta_1 }{\max\set{0, p_1'-\delta_1} \cdot e^{-\eps_1} (1-\gamma) + \gamma} \\
    \le & \frac{e^{2\eps_1} \cdot  p' + e^{\eps_1}\delta_1  }{\max\set{0,p_1'-\delta_1}\cdot  (1-\gamma) + \gamma} .
  \end{align*}
  Furthermore we have the following lowerbound:
  \begin{align*}
    \Pr\inb{ \Aout(\Data') \in F } =& p'\frac{1 -\inp{(1-p_1')(1-\gamma)}^T }{1-(1-p_1')(1-\gamma)} \\
    \ge & p' \frac{1-(1-\gamma)^T }{p_1'(1-\gamma) + \gamma} \\
    \ge & p' \frac{1-\eps_0/2 }{p_1'(1-\gamma) + \gamma}, \quad\hbox{ if $T\ge \frac{1}{\gamma}\ln \frac{2}{\eps_0}$}\\
    \ge & \frac{e^{-\eps_0} \cdot p' }{p_1'(1-\gamma) + \gamma}, \quad\hbox{ if $\eps_0 \le 1$}\\
  \end{align*}
  Then for an event $F = E\setminus\set{\bot}$ (that is, $F$ does not contain $\bot$), we have
  \begin{align*}
	  &\Pr\inb{ \Aout(\Data) \in F } - e^{2 \eps_1 + \eps_0} \cdot \Pr\inb{ \Aout(\Data') \in F } \\
  \le &e^{2\eps_1}p' \cdot \inp{\frac{1 }{\max\set{0,p_1'-\delta_1}\cdot (1-\gamma) + \gamma} - \frac{1 }{p_1'(1-\gamma) + \gamma} }+ \frac{e^{\eps_1} \delta_1}{\max\set{0,p_1'-\delta_1} (1-\gamma) + \gamma} \\
  \le & \frac{e^{2\eps_1} \delta_1 p' (1-\gamma) + e^{\eps_1}\delta_1 \inp{p_1'(1-\gamma) + \gamma }}{\inp{ \max\set{0,p_1'-\delta_1} (1-\gamma) + \gamma}  \inp{p_1'(1-\gamma) + \gamma }} \\
  \le & \frac{e^{2\eps_1} \delta_1 \inp{ 2p_1'(1-\gamma) + \gamma }}{\inp{ \max\set{0,p_1'-\delta_1} (1-\gamma) + \gamma}  \inp{p_1'(1-\gamma) + \gamma }}, \quad\hbox{by $p' \le p_1'$}\\
  \le &\frac{2e^{2\eps_1} \delta_1 }{\gamma}.
  \end{align*}

  Next, notice that we also have $1-p_1 \le e^{\eps_1} (1-p_1') + \delta_1$, then for the output $\bot$ we can upperbound
  \begin{align*}
	  \Pr\inb{ \Aout(\Data) = \bot } 
    =&  (1-p_1)\gamma \cdot \frac{1 +  \frac{p_1}{\gamma}(1-p_1)^{T-1}(1-\gamma)^T }{p_1(1-\gamma) + \gamma}\\
    \le& (1-p_1)\gamma \cdot \frac{1 +  \frac{1}{e(T-1)\gamma}(1-\gamma)^T }{p_1(1-\gamma) + \gamma}  , \quad \hbox{ by AM-GM inequality} \\
    \le& \inp{e^{\eps_1}(1-p_1') + \delta_1} \gamma\cdot \frac{1 +  \frac{1}{e(T-1)\gamma}(1-\gamma)^T }{e^{-\eps_1}\max\set{0,p_1'-\delta_1}(1-\gamma) + \gamma}, \quad \hbox{ by $(\eps_1,\delta_1)$-DP}  \\
    \le& \inp{e^{\eps_1}(1-p_1') + \delta_1} \gamma\cdot \frac{1 +  \eps_0/2 }{e^{-\eps_1}\max\set{0,p_1'-\delta_1}(1-\gamma) + \gamma}, \quad \hbox{ by the choice of $T$} \\
    \le&  \frac{e^{2\eps_1 + \eps_0}(1-p_1')\gamma + e^{\eps_1 + \eps_0} \delta_1\gamma }{\max\set{0,p_1'-\delta_1}(1-\gamma) + \gamma}.
  \end{align*}
  And we lowerbound
  \begin{align*}
  \Pr\inb{ \Aout(\Data') = \bot } 
    =&  (1-p_1')\gamma \cdot \frac{1 +  \frac{p_1'}{\gamma}(1-p_1')^{T-1}(1-\gamma)^T }{p_1'(1-\gamma) + \gamma}\\
    \ge&  (1-p_1')\gamma \cdot \frac{1 }{p_1'(1-\gamma) + \gamma}\\
  \end{align*}
  Therefore we have 
  \begin{align*}
    &\Pr\inb{ \Aout(\Data) = \bot } - e^{2\eps_1 + \eps_0} \cdot \Pr\inb{ \Aout(\Data') = \bot }\\
    \le & e^{2\eps_1 + \eps_0} (1-p_1')\gamma \inp{\frac{1 }{\max\set{0,p_1'-\delta_1} (1-\gamma) + \gamma} - \frac{1 }{p_1'(1-\gamma) + \gamma} } + \frac{ e^{\eps_1 + \eps_0} \delta_1\gamma }{\max\set{0,p_1'-\delta_1}(1-\gamma) + \gamma} \\
    = &  \frac{e^{2\eps_1 + \eps_0} \delta_1\gamma \cdot (1-p_1')(1-\gamma)}{\inp{\max\set{0,p_1'-\delta_1} (1-\gamma) + \gamma}\inp{p_1'(1-\gamma) + \gamma}}  + \frac{ e^{\eps_1 + \eps_0} \delta_1\gamma }{\max\set{0,p_1'-\delta_1}(1-\gamma) + \gamma} \\
    \le &  \frac{e^{2\eps_1 + \eps_0} \delta_1 \gamma  \cdot \inp{(1-p_1')(1-\gamma) + p_1'(1-\gamma) + \gamma}}{\inp{\max\set{0,p_1'-\delta_1} (1-\gamma) + \gamma}\inp{p_1'(1-\gamma) + \gamma}} \\ 
    \le & \frac{e^{2\eps_1 + \eps_0}\cdot \delta_1}{\gamma}.
  \end{align*}

  Finally, for an event $E$ that contains $\bot$, we let $F = E\setminus\set{\bot}$, and then
  \begin{align*}
	  &\Pr\inb{ \Aout(\Data) \in E } - e^{2 \eps_1 + \eps_0} \cdot \Pr\inb{ \Aout(\Data') \in E } \\
  =&\Pr\inb{ \Aout(\Data) \in F } + \Pr\inb{\Aout(\Data) = \bot} - e^{2 \eps_1 + \eps_0} \cdot \inp{\Pr\inb{ \Aout(\Data') \in F }+ \Pr\inb{\Aout(\Data') = \bot}} \\
  \le &\frac{2e^{2\eps_1} \delta_1 }{\gamma} + \frac{e^{2\eps_1 + \eps_0}\cdot \delta_1}{\gamma} \\
  \le &3e^{2\eps_1 + \eps_0}\cdot\frac{\delta_1}{\gamma}.
  \end{align*}

%

  \bigskip

  For part (d), notice that in each iteration, in order to not halt, $q$ has to be below $\T$, and the $\gamma$-biased coin test did not pass.
  In other words, for each iteration, $\Pr[\hbox{halting in any iteration}] = 1 - (1-p_1)(1-\gamma) = p_1 (1-\gamma) + \gamma$.
  Therefore this can be stochastically dominated by a geometric distribution (which corresponds to setting $T=\infty$), with expected number of trials being at most $\frac{1}{p_1(1-\gamma)+\gamma}$.

  \bigskip

  For part (e), by direct calculations,
  \begin{align*}
    \Pr\inb{\hbox{output $\bot$}} =& (1-p_1)\gamma \cdot \frac{1 +  \frac{p_1}{\gamma}(1-p_1)^{T-1}(1-\gamma)^T }{p_1(1-\gamma) + \gamma} \\
    \le& (1-p_1)\gamma \cdot \frac{1+ \frac{1}{e (T-1) \gamma} (1-\gamma)^T}{p_1} , \quad \hbox{ by AM-GM inequality} \\
    \le& \frac{(1-p_1) (1+ \eps_0/2)}{p_1}\gamma.
  \end{align*}
\end{proof}

\subsection{Proof of~\cref{thm:maxRandDPstop}}
\label{sec:proof-maxRandDPstop}

	For convenience we restate~\cref{thm:maxRandDPstop}.
\begin{theorem}
	Fix any $\eps_0 \in (0, 1/2), \gamma \in [0,1], \delta_2 > 0$ and let $T=\left\lceil \frac{1}{\gamma}  \inp{\ln \frac{2(1+\gamma)^2}{\eps_0 \gamma^2} + \ln \ln \frac{2(1+\gamma)^2}{\eps_0 \gamma^2}}\right\rceil$. Consider a variant of~\cref{alg:maxRand} that outputs the highest scored candidate from $S$ if $j$ reaches $T$.
	If $Q$ is $\eps_1$-DP, then the output of this algorithm is $(3\eps_1 + 3\eps_0)$-DP. 
\end{theorem}
\begin{proof}
	Similar to the proof of~\cref{thm:maxRandDP}: let $\Aout(\Data)$ be the output of~\cref{alg:maxRand} on $\Data$,  then we have
  \begin{align*}
	  &\Pr\inb{ \Aout(\Data) = (x,q) } \\
	  =& \sum_{j=1}^{T}\Pr\inb{ \Aout(\Data) = (x,q) \wedge \abs{S} = j } \\
	  =&\sum_{j=1}^T \Pr\inb{\abs{S} = j} \cdot \Pr\inb{\hbox{$\max S \le q$, and $(x,q) \in S$ }\mid \abs{S} = j} \\
	  =&\sum_{j=1}^T \inp{1-\gamma}^{j-1} \gamma \cdot \Pr\inb{\hbox{$\max S \le q$, and $(x,q) \in S$ }\mid \abs{S} = j} \\
	  & \quad + \inp{1-\gamma}^{T} \cdot \Pr\inb{\hbox{$\max S \le q$, and $(x,q) \in S$ }\mid \abs{S} = T}.
  \end{align*}
  Then, observe that
  \begin{align*}
	  \Pr\inb{\hbox{$\max S \le q$} \mid \abs{S}=j} = (1-p_0)^j,
  \end{align*}
  and
  \begin{align*}
	  \Pr\inb{ (x,q) \in S \mid \hbox{$\max S \le q$, and $\abs{S}=j$}} = 1 - \inp{1-\frac{p}{1-p_0}}^j= 1 - \inp{\frac{1-p_1}{1-p_0}}^j.
  \end{align*}
  Together we have
  \begin{align}
	  \begin{split}
    \Pr\inb{ \Aout(\Data) = (x,q) }
		  =  & \sum_{j=1}^T \inp{1-\gamma}^{j-1} \gamma \cdot  \inp{(1-p_0)^j - (1-p_1)^j} \\&\;\;+ \inp{1-\gamma}^{T} \inp{(1-p_0)^T - (1-p_1)^T}.
	  \end{split}
	  \label{eq:prob_aout}
  \end{align}
  We denote $a \defeq (1-\gamma)(1-p_0)$, $b \defeq  (1-\gamma)(1-p_1)$, then
  \begin{align*}
    &\Pr\inb{ \Aout(\Data) = (x,q) }\\
    =&\frac{\gamma(1-p_0)\inp{1 - a^T}}{1 - a} - \frac{\gamma(1-p_1)\inp{1 - b^T}}{1 - b} + a^T - b^T\\
    =&\frac{\gamma p + \inp{(1-p_1)(1-p_0)\gamma^2 - p_0 p_1 \gamma} (b^T - a^T) - \gamma(p_1 a^T - p_0 b^T)}{\inp{(1-p_0)\gamma+p_0}\inp{(1-p_1)\gamma+p_1}} + a^T - b^T .
  \end{align*}
  Observe that $a-b = (1-\gamma)p$, and we have
  \[
	  a^T - b^T = \sum_{i=0}^{T-1} a^{T-i} b^i - a^{T-i-1} b^{i+1} = (a-b)\sum_{i=0}^{T-1} a^{T-i-1} b^i \le T (a-b) a^{T-1}.  
	  \]
 Using the upper bound on $a$, this also implies that
	\[
		a^T - b^T \le T p (1-\gamma)^{T}.
		\]
  Furthermore,
  \begin{align*}
	  \abs{p_1 a^T - p_0 b^T} \le& \abs{p_1 a^T - p_0 a^T} + \abs{p_0 a^T - p_0 b^T} \\
	  \le& p a^T  + T p  (1-\gamma)^T \cdot p_0 (1-p_0)^{T-1} \\
	  \le& p (1-\gamma)^T  + \frac{T}{e (T-1)} p  (1-\gamma)^T, \hbox{ by AM-GM inequality}  \\
	  \le& 2 p (1-\gamma)^{T}.
  \end{align*}
  And we also have
  \begin{align*}
	  \inp{(1-p_0)\gamma+p_0}\inp{(1-p_1)\gamma+p_1} &\le (1+\gamma)^2\\
	  \abs{(1-p_1)(1-p_0)\gamma^2 - p_0 p_1 \gamma} &\le \gamma(1+\gamma).
  \end{align*}

  Now, if $T \ge \frac{1}{\gamma}  \inp{\ln \frac{2(1+\gamma)^2}{\eps_0 \gamma^2} + \ln \ln \frac{2(1+\gamma)^2}{\eps_0 \gamma^2}}$, then we have $T (1-\gamma)^T \le \frac{\eps_0\gamma}{(1+\gamma)^2}$.
  Therefore we can upperbound
  \begin{align*}
    &\Pr\inb{ \Aout(\Data) = (x,q) }\\
    =&\frac{\gamma(1-p_0)\inp{1 - a^T}}{1 - a} - \frac{\gamma(1-p_1)\inp{1 - b^T}}{1 - b} + a^T - b^T\\
    \le&\frac{\gamma(1-p_0)}{1 - a} - \frac{\gamma(1-p_1)}{1 - b} + \frac{\eps_0 \gamma p}{(1+\gamma)^2}\\
    \le&\frac{\gamma p (1+ \eps_0) }{\inp{(1-p_0)\gamma+p_0}\inp{(1-p_1)\gamma+p_1}} .
  \end{align*}
The first inequality above is a consequence of upper bounding the sum of the first $T$ terms in~\cref{eq:prob_aout} by the sum to infinity.
  Then we lowerbound
  \begin{align*}
    &\Pr\inb{ \Aout(\Data) = (x,q) }\\
    =&\frac{\gamma p + \inp{(1-p_1)(1-p_0)\gamma^2 - p_0 p_1 \gamma} (b^T - a^T) - \gamma(p_1 a^T - p_0 b^T)}{\inp{(1-p_0)\gamma+p_0}\inp{(1-p_1)\gamma+p_1}} + a^T - b^T \\
    \ge&\frac{\gamma p - \frac{\eps_0\gamma^2}{1+\gamma}p - \frac{2\gamma^2\eps_0}{T (1+\gamma)^2}p}{\inp{(1-p_0)\gamma+p_0}\inp{(1-p_1)\gamma+p_1}} \\
    \ge&\frac{\gamma p - \frac{ \gamma + 3\gamma^2 }{(1+\gamma)^2}\eps_0 \gamma p}{\inp{(1-p_0)\gamma+p_0}\inp{(1-p_1)\gamma+p_1}} \\
    \ge&\frac{\gamma p (1- \eps_0) }{\inp{(1-p_0)\gamma+p_0}\inp{(1-p_1)\gamma+p_1}}.
  \end{align*}

  Finally, 
  \begin{align*}
    \frac{\Pr\inb{ \Aout(\Data) = (x,q) } }{\Pr\inb{ \Aout(\Data') = (x,q) } }
    \le \frac{p}{p'}\cdot \frac{p_0'(1-\gamma) + \gamma}{p_0(1-\gamma) + \gamma} \cdot \frac{p_1'(1-\gamma) + \gamma}{p_1(1-\gamma) + \gamma} \cdot \frac{1+\eps_0}{1-\eps_0}
    \le  \exp(3\eps_1 + 3 \eps_0) .
  \end{align*}
\end{proof}

\subsection{Proof of~\cref{lem:trivial-coupling}}
\label{sec:proof-trivial-coupling}

We re-state~\cref{lem:trivial-coupling} below for convenience.
\begin{lemma}
	Let $\set{X_1, \cdots, X_n}$ and $\set{Y_1, \cdots, Y_n}$ be two sequences of independent $\set{0,1}$ random variables, and let $X = \sum_{i=1}^n X_i$, $Y= \sum_{i=1}^n Y_i$.
	For any fixed $\eps_1 \in (0,1), \eps_0 \in (0,1)$, $\delta_0 \in (0,1)$,
let $C=2(e^{\eps_0+\eps_1} +1 + e^{\eps_0/2}) < 21$.

If $\E X \le e^{\eps_1} \E Y $, then under the trivial (independent) coupling between $X$ and $Y$,
	\[
		\Pr\inb{ X \ge e^{\eps_1+\eps_0} \cdot Y +  \frac{C}{\eps_0} \cdot \ln \frac{2}{\delta_0}} \le \delta_0.
	\]
	Equivalently, if we let $\Delta \defeq \frac{C\ln \frac{2}{\delta_0}}{\eps_0 \inp{e^{\eps_0+\eps_1} - 1}} = O\inp{\frac{1}{\eps_0^2} \ln\frac{1}{\delta_0}}$, then
	\[
	  \Pr\inb{ X+ \Delta  \ge e^{\eps_1+\eps_0} \cdot \inp{ Y +  \Delta}} \le \delta_0.
	\]
\end{lemma}
Before proving this lemma, it will be useful to show the following concentration bounds.
\begin{lemma}
	Let $X$ be a sum of independent $\set{0,1}$ random variables: $X = \sum_{i=1}^n X_i$ as defined in~\cref{lem:trivial-coupling}, then $\forall \eps \in (0,1), \delta \in (0,1)$,
	\begin{align}
		\Pr\inb{ X \ge e^{\eps} \E X + \frac{e^{\eps} + 1}{\eps}\ln \frac{1}{\delta} } \le \delta.
		\label{eq:concentrationX}
	\end{align}
	\label{lem:concentrationX}
\end{lemma}
\begin{proof}

	Let $\Delta_1 = \frac{e^{\eps} + 1}{(e^{\eps} - 1)^2}\ln \frac{1}{\delta}$, then we apply the standard Chernoff bound to $X+\Delta_1$:
	\begin{align*}
		\Pr\inb{X+\Delta_1 \ge e^{\eps} \cdot \E \inp{ X + \Delta_1}}\le& \exp\inp{\frac{-(e^{\eps} -1)^2 \E(X+\Delta_1)}{e^{\eps} + 1}}
		\le\exp\inp{\frac{-(e^{\eps} -1)^2 \Delta_1}{e^{\eps} + 1}}
		=\delta.
	\end{align*}
	By re-arranging, we get that
	\[
		\Pr\inb{ X \ge e^{\eps} \E X + \frac{e^{\eps} + 1}{\eps}\ln \frac{1}{\delta} } \le \Pr\inb{ X \ge e^{\eps} \E X + \frac{e^{\eps} + 1}{e^{\eps} - 1}\ln \frac{1}{\delta} } = \Pr\inb{X\ge e^{\eps} \E X + (e^{\eps} - 1)\Delta_1} \le \delta.
	\]
\end{proof}
\begin{lemma}
	Let $Y$ be a sum of independent $\set{0,1}$ random variables: $Y = \sum_{i=1}^n X_i$ as defined in~\cref{lem:trivial-coupling}, then $\forall \eps \in (0,1), \delta \in (0,1)$,
	\begin{align}
		\Pr\inb{Y \le e^{-\eps} \E Y - \frac{\ln\frac{1}{\delta}}{\eps} } \le \delta.
		\label{eq:concentrationY}
	\end{align}
	\label{lem:concentrationY}
\end{lemma}
\begin{proof}
	Note that by a direct application of Chernoff bound, it holds that $\forall \delta \in (0,1)$,
	\begin{align*}
		\Pr\inb{ Y \le \inp{1- \sqrt{\frac{2\ln\frac{1}{\delta}}{\E Y}}} \cdot \E Y  } \le \exp\inp{ -\inp{\tfrac{2\ln\frac{1}{\delta}}{\E Y}}\cdot \frac{\E Y}{2}} = \delta.
	\end{align*}

	Then, by AM-GM inequality: $\forall \eps$, we have $\frac{\eps \E Y}{2}  + \frac{\ln\frac{1}{\delta}}{\eps} \ge \sqrt{2\E Y \cdot \ln\frac{1}{\delta}}$. 
	Moreover, by standard estimates: for $\eps \in (0,1)$, we have $1- \eps \le e^{-\eps} \le 1- \eps/2$.
	Therefore, $\forall \eps \in (0,1),\delta \in (0,1)$,
	\begin{align*}
		\Pr\inb{Y \le e^{-\eps} \E Y - \frac{\ln\frac{1}{\delta}}{\eps} } 
		\le &\Pr\inb{ Y \le  \E Y - \inp{\frac{\eps \E Y}{2} + \frac{\ln\frac{1}{\delta}}{\eps} }} \\
		\le &\Pr\inb{ Y \le  \E Y - \sqrt{2\E Y \cdot \ln\frac{1}{\delta}}} , \quad \hbox{ by AM-GM inequality} \\
		= &\Pr\inb{ Y \le \inp{1- \sqrt{\frac{2\ln\frac{1}{\delta}}{\E Y}}} \cdot \E Y  } \\
		\le &\delta.
	\end{align*}

\end{proof}

Finally, we are ready to prove~\cref{lem:trivial-coupling}.

{\noindent Proof of~\cref{lem:trivial-coupling}.}
	For any given $\eps_0, \delta_0$, we set $\eps=\eps_0/2$, and $\delta=\delta_0/2$.
	Then we consider the following events, $G_X$ and $G_Y$ on the probability space of $X$ and $Y$ respectively:
	\begin{align*}
	G_X \defeq& \set{X: X < e^{\eps} \E X + \frac{e^{\eps} + 1}{\eps}\ln \frac{1}{\delta} };\\
	G_Y\defeq& \set{Y: Y > e^{-\eps} \E Y - \frac{\ln\frac{1}{\delta}}{\eps}}.
	\end{align*}
	As discussed in~\cref{lem:concentrationX,lem:concentrationY}, we have
	\[
		\Pr\Big[ \overline{G_X} \cup \overline{G_Y} \Big] \le \Pr\Big[\overline{G_X}\Big] + \Pr\Big[\overline{G_Y}\Big]
		\le \delta + \delta = \delta_0.
	\]

	On the other hand, conditional on $G_X$ and $G_Y$, we must have
	\begin{align*}
		X <& e^{\eps} \E X + \frac{e^{\eps} + 1}{\eps}\ln \frac{1}{\delta} , \quad\hbox{ by $G_X$} \\
		\le& e^{\eps+\eps_1} \E Y + \frac{e^{\eps} + 1}{\eps}\ln \frac{1}{\delta}  \\
		\le& e^{2\eps+\eps_1} \inp{Y + \frac{\ln\frac{1}{\delta}}{\eps}} +\frac{e^{\eps} + 1}{\eps}\ln \frac{1}{\delta}, \quad\hbox{ by $G_Y$} \\
		=& e^{\eps_0+\eps_1} \cdot Y +  \frac{2(e^{\eps_0+\eps_1} +1 + e^{\eps_0/2})  }{\eps_0}\ln\frac{2}{\delta_0} \\
	\end{align*}
	Therefore, let $C=2(e^{\eps_0+\eps_1} +1 + e^{\eps_0/2})$, then
	\[
		\Pr\inb{ X \ge e^{\eps_0+\eps_1} \cdot Y +  \frac{C \cdot \ln\frac{2}{\delta_0}}{\eps_0} } \le 
		\Pr\Big[ \overline{G_X} \cup \overline{G_Y} \Big] \le \delta_0.
	\]
\qed

%

%% file: naive_algs.tex
\section{Naive algorithms: tight examples and analysis}
\subsection{Outputting the best candidate}
  \label{sec:naiveDP}
  In this subsection we consider a naive algorithm where, one simply chooses the best candidate (with the highest score, e.g., in the hyperparameter selection setting, among the trained models one outputs the best performing model and its corresponding hyperparameter). 
  
  What is the best $\eps$-DP bound, or $(\eps,\delta)$-DP bound that we can hope for? 
  Basic composition theorem says that if there are $K$ candidates, and each candidate is $\eps$-DP, then, outputting the best of the $K$ candidates is $(K \eps)$-DP. This is actually tight, thanks to the following example.
    \begin{align*}
      \hbox{let $m \sim M_i(d_1)$, and } q_i(m) &= \begin{cases}
	0.9, \hbox{ if $i=0$} \\
	0.8, \hbox{ if $i\neq 0$ and with probability $\frac{1}{2}$ } \\
	0.95,\hbox{ if $i\neq 0$ and with probability $\frac{1}{2}$ } 
    \end{cases} \\
            \hbox{let $m \sim M_i(d_1')$, and } q_i(m) &= \begin{cases}
	0.9, \hbox{ if $i=0$} \\
	0.8, \hbox{ if $i\neq 0$ and with probability $\frac{e^{\eps}}{2}$ } \\
	0.95,\hbox{ if $i\neq 0$ and with probability $\frac{1-e^{\eps}}{2}$ } 
    \end{cases}
    \end{align*}
    Here the probability are with respect to the randomness in the $\eps$-DP candidate $M$.
    Then for neighboring datasets $d_1$ and $d_1'$, we get $K$ samples of the candidates (e.g. for each of the $K$ candidates, we draw a sample), and then we compare the event of choosing $i=0$ as the best hyperparameter.
    It is easy to see that
    $\Pr\inb{i_*(d_1) = 0} = 2^{-K}$ and $\Pr\inb{i_*(d_1') = 0} = \exp\inp{K\eps} \cdot 2^{-K}$, therefore,
  \begin{align*}
    \ln \abs{\frac{\Pr\inb{i_*(d_1') = 0}}{\Pr\inb{i_*(d_1) = 0}}}  = K \eps.
  \end{align*}

  What about $(\eps,\delta)$-DP bound? 
  We show that outputting the maximum cannot do better than $\inp{\Theta(\ln\frac{1}{\delta}) \eps, \delta}$-DP.
  Fix an integer $K$, and $\delta\in (0,1)$.

    \begin{align*}
      \hbox{let $m \sim M_i(d_1)$, and } q_i(m) &= \begin{cases}
	0.9, \hbox{ if $i=0$} \\
	0.8, \hbox{ if $i\neq 0$ and with probability $1-\frac{\ln 1/\delta}{K}$ } \\
	0.95,\hbox{ if $i\neq 0$ and with probability $\frac{\ln 1/\delta}{K}$ } 
    \end{cases} \\
            \hbox{let $m \sim M_i(d_1')$, and } q_i(m) &= \begin{cases}
	0.9, \hbox{ if $i=0$} \\
	0.8, \hbox{ if $i\neq 0$ and with probability $1-\frac{e^{-\eps}\ln 1/\delta}{K}$ } \\
	0.95,\hbox{ if $i\neq 0$ and with probability $\frac{e^{-\eps}\ln 1/\delta}{K}$ } 
    \end{cases}
    \end{align*}
    Again for neighboring datasets $d_1$ and $d_1'$, we get $K$ samples of the candidates (e.g. for each of the $K$ candidates, we draw a sample), and then we compare the event of choosing $i=0$ as the best hyperparameter.
    It is easy to see that
    $\Pr\inb{i_*(d_1) = 0} \approx \delta$ and $\Pr\inb{i_*(d_1') = 0} \approx \delta^{1-\eps}$, therefore,
  \begin{align*}
    \ln \abs{\frac{\Pr\inb{i_*(d_1') = 0} - \delta}{\Pr\inb{i_*(d_1) = 0}}} \approx \ln \inp{\delta^{-\eps} - 1} \approx \inp{\ln 1/\delta} \eps.
  \end{align*}
  \CommentS{
  \begin{theorem}
    Let $M_1, M_2, \cdots, M_K$ be $\eps$-DP candidates.
    Then outputting the candidate with the highest score is $(\Theta(\ln \frac{1}{\delta})\eps, \delta)$-DP for every $\delta>0$.
    And this bound is also tight for every $\delta$.
    \label{thm:eps-delta-DP}
  \end{theorem}
  \begin{proof}
    TBA.
  \end{proof}
  }

\subsection{Thresholding with decreasing thresholds}
In this subsection we consider a natural variant of~\cref{alg:thresholding}: 
in each iteration, instead of halting (and output $\perp$) with probability $\gamma$, 
what if we decrease the threshold? 
In particular, we will try a lower threshold with probability at least $\gamma$ in each step.
Is this good enough, so that we can avoid paying the privacy cost for the different thresholds that we tried along the way?
See~\cref{alg:decremental-thresholding} for formal description.
For simplicity, we consider the special case where we do not stop the algorithm after some finite number of $T$ steps. The algorithm could run forever \emph{in the worst case}. Note that in~\cref{alg:thresholding}, running the algorithm longer only helps in privacy (recall that in~\cref{thm:thresholdingDP}, the larger $T$ is, the smaller $\eps_0$ we can choose).

  \begin{algorithm}
    \flushleft Input: a budget $\gamma \le 1$, an integer $R$ for how many thresholds to try, and the sampling access to $Q(\Data)$.

    Let $\T = 1$;

  While $\T\ge 0$:
  \begin{itemize}
    \item draw $(x,q) \sim Q(\Data)$;
    \item if $q \ge \T$ then output $(m,i)$ and halt; 
    \item flip a $\gamma$-biased coin: with probability $\gamma$, set $\T \gets \T - \frac{1}{R}$.
  \end{itemize}
  \caption{Thresholding with decreasing thresholds.}
  \label{alg:decremental-thresholding}
\end{algorithm}

Note that as soon as $\T=0$, the algorithm will output whichever samples of candidate that it gets, as $q \ge \T$ is trivially true.

Here is an example which shows that trying many thresholds are not free for privacy: if we plan to try $R$ thresholds, then we do have to pay a factor of $R$ in the privacy cost.

Consider $Q(\Data) = \mathrm{Bernoulli}(p)$, that is, $q\sim Q(\Data)$ will be $1$ with probability $p$, and $0$ otherwise. 
Therefore, in order for the algorithm to ouput a candidate with score $0$, the threshold has to be decreased $R$ times, untill $\T=0$. 
And only then~\cref{alg:decremental-thresholding} will output a score $0$ candidate with probability $1-p$.
Let us compute the probability of this event.
Let $a = (1-p) (1-\gamma), b =(1-p)\gamma$.

\begin{align*}
  \Pr\inb{\hbox{output $1$ before the $R$-th decrements on $\Data$}} 
  =& p \sum_{j=0}^{\infty} \sum_{i=0}^{R-1} {j \choose i} a^{j-i} b^i \\
  =&p \sum_{i=0}^{R-1} b^i \sum_{j=i}^\infty {j \choose i} a^{j-i} \\
  =&p \sum_{i=0}^{R-1} b^i \frac{1}{(1-a)^{i+1}} \\
  =& \frac{p}{1-a}\frac{1 - \inp{\frac{b}{1-a}}^R}{1 - \frac{b}{1-a}} \\
  =& 1 - \inp{\frac{b}{1-a}}^R.
\end{align*}
Therefore 
\begin{align*}
  &\Pr\inb{\hbox{output $0$ on $\Data$}} \\
  = &(1-p) \cdot \inp{ 1 -\Pr\inb{\hbox{output $1$ before the $R$-th decrements on $\Data$}}} \\
  = &(1-p)\inp{\frac{b}{1-a}}^R = (1-p) \gamma^R \cdot \inp{\frac{1-p}{p(1-\gamma) + \gamma}}^R. 
\end{align*}

Notice that $\frac{1-p}{p(1-\gamma) + \gamma}$ is monotone in $p$, if $(1-p)$ changes to $e^{\eps} (1-p)$, then the $e^{\eps}$ factor will be amplified $R$ times.


\subsection{Outputting the $p$-th percentile}
\label{sec:percentileDP}
Without loss of generality, we consider $p=\frac{1}{2}$, that is, we output the median candidate.
Also without loss of generality, let us say there are only two models, $m_1$ and $m_2$, and $q(m_1)=0$, $q(m_2)=1$.
Consider the following two distributions of $M_i(\Data)$ and $M_i(\Data')$.

    \begin{align*}
      \forall i, M_i(\Data)&= \begin{cases}
	  m_1, \hbox{ with probability $\frac{1-\eps}{2}$ } \\
	  m_2,\hbox{ with probability $\frac{1+\eps}{2}$ } 
	\end{cases}, \\
    \forall i, M_i(\Data') &= \begin{cases}
	      m_1,\hbox{ with probability $\frac{1+\eps}{2}$ } \\
	      m_2,\hbox{ with probability $\frac{1-\eps}{2}$ } 
	    \end{cases}.
    \end{align*}

    Clearly, the two distributions are $O(\eps)$-close, yet in one distribution, the median is $m_2$, while in the other the median is $m_1$.
    Therefore, the median of the distribution is not private. This is also the case if one takes the median of $N$ samples, assuming $N$ large enough (where we have concentration with high probability).
    This also applies if one pick an index $k$ from $\set{1,2,\cdots, \lceil N/2 \rceil}$ uniformly at random, and then output the $k$-th highest.

%% file: priv_amp.tex
\section{Improved analysis of the private amplification algorithm in~\cite{GLMRT}}

Let $\set{Q_i(\Data)}_{i=1}^N$ be a sequence of independent distributions, let $\set{q_i(\Data)}_{i=1}^N$ be the random variables where $q_i \sim Q_i$.

Suppose that for some $\eps_1$, every $q_i$ is $\eps_1$-DP.
For a given $\T \in \R$, let $\qt_i = \min\set{\T, q_i}$.
Denote $\EM_\eps(A_1,\cdots,A_n)$ to be the exponential mechanism (with parameter $\eps$) on the sequence $\set{A_i}$, which is a random variable. 
For any given $\gamma>0$, $n=N+ 1 + \frac{1}{\gamma}$, and $A_i =
\begin{cases}
	\qt_i(\Data), & \hbox{ if $i\le N$} \\
	\T, & \hbox{ otherwise}
\end{cases}$.
Then, we consider the distribution of $\EM_{\eps_2}(A_1, \cdots, A_n)$.
Basically, we will add $1+\frac{1}{\gamma}$ dummy classes with a score $\T$, scale back everything else and then apply exponential mechanism.

\begin{theorem}
	Fix $\eps_2$, $\T$, $\delta$, $\gamma \in (0,\frac{1}{4})$.
  Suppose that for every $i$, $q_i$ is $\eps_1$-differentially private.
  Let $p= \frac{1}{N} \sum_{i=1}^N \E \exp\inp{\eps_2(\widetilde{q_i} - \T)}$.
  Then the following holds.
  \begin{itemize}
    \item Utility: 
      The mechanism outputs a dummy class with probability $\frac{\gamma + 1}{Np\gamma + 1}$, and
      \[
	      \Pr\inb{\EM_{\eps_2} (A_1, \cdots, A_n) \ge \T - \frac{1}{\eps_2} \ln \left( \frac{1}{\delta p} \right) } \ge 1 - \delta,
	\]
	where the randomness is over both the internal randomness of exponential mechanism and the randomness of $\set{q_i}$.

    \item Privacy:
	    $\EM_{\eps_2}(A_1, \cdots, A_n)$ is $(2\eps_1 + 8\gamma)$-DP.

  \end{itemize}
  \label{thm:em}
\end{theorem}

It is worth noting that the privacy on the training set does not depend on $\eps_2$. In other words, one can even set $\eps_2 \to \infty$, which corresponds to sampling uniformly from the classes with a score exceeding some threshold and with one extra dummy class. The theorem says that doing so does not compromise the privacy of the training set at all.

Before we prove the theorem, we introduce a useful lemma similar to that of~\cite[Lemma C.1]{mcgregor2010limits}.
The key changes will be from a Binomial distribution to one that takes value from $[0,1]$.
\begin{lemma}
  Let $\set{X_i}_{i=1}^N$ be a sequence of independent random variables over $[0,1]$, then
  \[
    \frac{1}{1 + \sum_{i=1}^N \E X_i} \le \E \inb{\frac{1}{1 + \sum_{i=1}^N X_i }} \le \frac{1}{\sum_{i=1}^N \E X_i}.
  \]
  \label{lem:inv-exp}
\end{lemma}
\begin{proof}
  The first inequality follows by Jensen's inequality, since the function $f(X) = \frac{1}{1+X}$ is convex for $X > -1$.

  Next we use a formula for negative moments~\cite{negative_moments}. Note that for every $u, x>0$,
  \[
    \frac{u^{1+x}}{1+x} = \Int{t^x}{t,0,u}. 
  \]

  Setting $u=1$ and taking expectations over $x=\sum X_i$, 
  \begin{align*}
    \E \inb{\frac{1}{1+\sum X_i} } =& \Int{\E\inb{t^{\sum_{i=1}^N X_i}}}{t,0,1} 
    =\Int{\prod_{i=1}^N \E\inb{t^X_i}}{t,0,1}.
  \end{align*}
  
  Next we show that $t^x \le (t-1)x+1$ for $x\in [0,1]$. 

  For any given $t$, consider the following two points: 
  $t^x=
  \begin{cases}
    1, \hbox{ if $x=0$} \\
    t, \hbox{ if $x=1$}
  \end{cases}$. Thus $(t-1)x + 1$ is the line joining these two points.
  Since $t^x = e^{x \ln t}$ is convex as long as $t>0$, and $t^x$ meets $(t-1)x + 1$ at the two points, we get that $t^x \le (t-1)x+1$ for $x \in [0,1]$. 

  Therefore,
  \begin{align*}
    \E \inb{\frac{1}{1+\sum X_i} }=&\Int{\prod_{i=1}^N \E\inb{t^X_i}}{t,0,1} \\
    \le &\Int{\prod_{i=1}^N (1+ (t-1)\E X_i)}{t,0,1} \\
      \le &\Int{\exp\inp{ (t-1)\sum_{i=1}^N \E X_i}}{t,0,1} \\
      =&\frac{1 - \exp\inp{\sum_{i=1}^N \E X_i}}{\sum_{i=1}^N \E X_i}.
  \end{align*}
  This concludes the proof.
\end{proof}

Now we are ready to prove~\cref{thm:em}.


{\bf Proof of~\cref{thm:em}.}
The probability of outputting a dummy is 
\[
	\E \inb{\frac{(1+\frac{1}{\gamma})\exp(\eps_2 \T)}{ (1+\frac{1}{\gamma})\exp(\eps_2 \T) + \sum_{i=1}^N \exp(\eps_2 \widetilde{q_i})} } \le \frac{1+\frac{1}{\gamma}}{Np+\frac{1}{\gamma}},
\]
where the inequality follows from the definition of $p$ and~\cref{lem:inv-exp}.
Then, 
\begin{align*}
  &\Pr\inb{\EM_{\eps_2}  (A_1, \cdots, A_n)< \T - \frac{1}{\eps_2} \ln \frac{1}{\delta p} }\\
  \le& \sum_{i \in [N]} \Pr\inb{ \widetilde{q_i} < \T - \frac{1}{\eps_2} \ln \frac{1}{\delta p}  } \exp(\eps_2 \widetilde{q_i}) \cdot \E \inb{ \frac{1}{ (1+\frac{1}{\gamma})\exp(\eps_2 \T) + \sum_{j=1}^N \exp(\eps_2 \widetilde{q_j})} \middle| \widetilde{q_i}} \\
  <& N p \delta \cdot \E \inb{\frac{1}{\frac{1}{\gamma}+ \sum_{j=1}^N \exp(\eps_2 (\widetilde{q_j} - \T))}} \\
  \le& N p \delta \cdot \frac{1}{\frac{1}{\gamma} - 1 + \sum_{i=1}^N \E \exp(\eps_2 (\widetilde{q_i} - \T))} \quad \hbox{ by~\cref{lem:inv-exp}}\\
    \le& \delta \quad \hbox{ by definition of $p$.}
\end{align*}

For the privacy part, consider any two neighboring datasets $\Data_1$ and $\Data_2$.
Let $A_i = \widetilde{q_i}(\Data_1)$, and $B_i=\widetilde{q_i}(\Data_2)$ for $i\le N$, and $A_i=B_i=\T$ for $i=N+1, \cdots, N+1+\frac{1}{\gamma}$.
Let $M_1 \defeq \EM_{\eps_2}\inp{ A_1, \cdots, A_n}$, $M_2\defeq \EM_{\eps_2}\inp{ B_1, \cdots, B_n }$ be the random variables of the two outcomes.

For any outcome $S$,
\begin{align*}
  \Pr[M_1 = S]
  =& \sum_{i=1}^N \Pr[A_i = S] \exp(\eps_2 S) \cdot \E \inb{\frac{1}{\exp(\eps_2 \T) + \sum_{j=1}^N \exp(\eps_2 A_j)} \middle| A_i = S} \\
  =& \sum_{i=1}^N \Pr[A_i = S] \exp(\eps_2 (S-\T)) \cdot \E \inb{\frac{1}{1 + \frac{1}{\gamma}+ \exp(\eps_2 (S-\T)) + \sum_{j=1,j \neq i}^N \exp(\eps_2 (A_j-\T))} }. 
\end{align*}

Note that by~\cref{lem:inv-exp},
\begin{align*}
  \E \inb{\frac{1}{1 + \frac{1}{\gamma}+ \exp(\eps_2 (S-\T)) + \sum_{j=1,j \neq i}^N \exp(\eps_2 (A_j-\T))} } 
  \le & \E \inb{\frac{1}{1 + \frac{1}{\gamma}+  \sum_{j=1,j \neq i}^N \exp(\eps_2 (A_j-\T))} } \\
  \le &\frac{1}{\frac{1}{\gamma} + \E \sum_{j=1,j \neq i}^N \exp(\eps_2 (A_j-\T))}  \\
  \le &\frac{1}{\frac{1}{\gamma} +\E \sum_{j=1}^N \exp(\eps_2 (A_j-\T))  - 1} \\
  \le &\frac{1}{Np + \frac{1}{\gamma}- 1}.
\end{align*}

On the other hand, by Jensen's inequality,
\begin{align*}
  &\E \inb{\frac{1}{1 + \frac{1}{\gamma} +\exp(\eps_2 (S-\T))+  \sum_{j=1,j \neq i}^N \exp(\eps_2 (A_j-\T))} } \\
  \ge &\frac{1}{1+ \frac{1}{\gamma} +\exp(\eps_2 (S-\T))+\E \sum_{j=1,j \neq i}^N \exp(\eps_2 (A_j-\T))}  \\
  \ge &\frac{1}{2 +\frac{1}{\gamma} + \E \sum_{j=1}^N \exp(\eps_2 (A_j-\T)) } \\
  \ge & \frac{1}{2+\frac{1}{\gamma} + Np}.
\end{align*}

By symmetry we have the same bound for $M_2$ and $B_j$.
Therefore,
\begin{align*}
  \frac{\Pr[M_1 = S]}{\Pr[M_2 = S]} \le& \max_{i \in [N]}\set{ \frac{ \Pr[A_i = S] \E \inb{\frac{1}{1 + \exp(\eps_2 (S-\T)) + \sum_{j=1,j \neq i}^N \exp(\eps_2 (A_j-\T))} }}{\Pr[B_i = S] \E \inb{\frac{1}{1 + \exp(\eps_2 (S-\T)) + \sum_{j=1,j \neq i}^N \exp(\eps_2 (B_j-\T))} }}}  \\
  \le&\exp(\eps_1)\frac{Np e^{\eps_1}+\frac{1}{\gamma} +2}{Np + \frac{1}{\gamma} - 1} \\
  \le& \exp\inp{2\eps_1 + 8\gamma}.
\end{align*}

Similarly
\begin{align*}
  \frac{\Pr[M_1 = S]}{\Pr[M_2 = S]} \ge& \min_{i \in [N]}\set{ \frac{ \Pr[A_i=S] \E \inb{\frac{1}{1 + \exp(\eps_2 (S-\T)) + \sum_{j=1,j \neq i}^N \exp(\eps_2 (A_j-\T))} }}{\Pr[B_i=S] \E \inb{\frac{1}{1 + \exp(\eps_2 (S-\T)) + \sum_{j=1,j \neq i}^N \exp(\eps_2 (B_j-\T))} }}}\\
  \ge&\exp(-\eps_1)\frac{Npe^{-\eps_1}+\frac{1}{\gamma}-1}{Np+\frac{1}{\gamma} + 2} \\
  \ge& \exp\inp{-2\eps_1 - 8\gamma }.
\end{align*}

This concludes the proof.

%% file: lower_bounds.tex
\section{Lower Bounds}
\label{app:lower_bounds}
In this section, we show that our algorithms loss in parameters are close to optimal. First note that since are competing against the best of $K$ mechanisms, at least $K$ oracle calls are needed, and our algorithm makes only $\tilde{O}(K)$ oracle calls.

When each of the input mechanisms $\M_i$ is $\eps$-DP, our final algorithm has privacy guarantee $2\eps + \eps'$ where $\eps'$ can be made arbitrarily small. Recall that in this factor of two loss occurs already in the case when $\M_i$ is $q_i(\cdot) + Lap(S/\eps)$ for some score functions with sensitivity $S$: in this case the NoisyMax mechanism has $2\eps$-DP, and the exponential mechanism with similar utility has the same factor of two loss. We next argue that this factor of two loss is necessary under weak utility assumptions. We start with a definition.

\begin{definition}
	Suppose that $\M$ is an algorithm that takes as input a set of $\eps$-DP mechanisms $\M_1, \ldots, \M_K$ and outputs an index $i$. We say that $i^*$ is $\gamma$-dominant in $\M_1, \ldots, \M_K$ on $D$ if $\Pr_{m_i \sim \M_i(d)}[\argmax_{i} m_i = i^*] \geq 1 - \gamma$. We say that $\M$ is $\gamma$-weakly useful if $\Pr[\M(D) = i^*] \geq \gamma$ whenever $i^*$ is $\gamma$-dominant in $\M_1, \ldots, \M_K$ on $D$.
\end{definition}

The next theorem says that a fairly mild weak usefulness condition already implies that this factor of $2$ loss is unavoidable.

\begin{theorem}
	Suppose that $\M$ is an algorithm that takes as input a set of $\eps$-DP mechanisms $\M_1, \ldots, \M_K$, and outputs an index $i$. If $\M$ is $\gamma$-weakly useful for $\gamma = K^{-\alpha}$ for a small enough $\alpha > 0$, then $\M$ cannot by $\hat{\eps}$-DP for any $\hat{\eps} < (2-6\alpha)\eps$.
\end{theorem}
\begin{proof}
	The proof is a simple packing argument. Our mechanisms $\M_i$ all have range $\{0, 1\}$ and output $1$ with probability $p_i(D)$ on dataset $D$. We define a set of $K+1$ datasets $D_0, D_1, \ldots, D_K$ such that:
	\begin{align*}
		p_i(D_0) &=  \frac{1}{2K^{0.5}}\\
		p_i(D_j) &= \left\{\begin{array}{ll}
			1-\frac{1}{2K^\alpha} & \mbox{if } i = j,\\
			\frac{1}{2K^{1+\alpha}} &\mbox{otherwise}.
		\end{array}\right.
	\end{align*}
	It is easy to check that if $D_0$ and $D_i$ are distance $\Delta = \lceil\frac{(0.5+\alpha)\ln K}{\eps}\rceil$, then the $\M_i$'s can be extended to satisfy $\eps$-DP. Moreover, any $\frac{1}{K^\alpha}$-weakly useful algorithm on dataset $D_i$ should output $i$ with probability at least $\frac{1}{K^{\alpha}}$. Suppose that $\M$ is $\hat{\eps}$-DP. Then,
	\begin{align*}
		\Pr[\M(D_0) = i] \geq \exp(-\Delta \hat{\eps}) \cdot \frac 1 {K^\alpha}.
	\end{align*}
	Since $\sum_i\Pr[\M(D_0) = i] \leq 1$, it follows that for some $i$, this probability $\Pr[\M(D_0) = i] \leq \frac 1 K$. It follows that
	\begin{align*}
		-\ln K &\geq -\Delta \hat{\eps} - \alpha\ln K \\
		\Leftrightarrow \;\;\;\;\;\;\;\hat{\eps} &\geq \frac{(1-\alpha) \ln K}{\lceil((0.5 + \alpha)\ln K)/\eps\rceil}. 
	\end{align*}
	For large enough $K$, this implies that $\hat{\eps} \geq 2(1-3\alpha)\eps$.
\end{proof}

%% file: app_dp_properties.tex
\section{Useful Properties of Differential Privacy}
\label{app:ed_v_cond_e}
In this section, we prove some folklore properties of the distance implicit in
the definition of differential privacy that are useful. We start with a definition of closeness.

\begin{definition}
  For distributions $P$ and $Q$, we say that $P$ is $(\eps,\delta)$-far from $Q$, if for all events $S$,
  \begin{align*}
    \Pr_{x \sim P}[x \in S] \leq \exp(\eps) \Pr_{x \sim Q}[x \in S] + \delta
  \end{align*}
  We say that $P \equiv_{\eps,\delta} Q$ if $P$ is $(\eps,\delta)$-far from $Q$ and $Q$ is $(\eps, \delta)$-far from $P$.
\end{definition}

\begin{lemma}
  \label{lem:ed_to_cond_e}
  Suppose that $P \equiv_{\eps,\delta} Q$ for $\delta < \frac 1 {10}$. Then for any $\eps' > \eps$, there is an event $B$ such that (a) $\Pr_{x \sim P}[x \in B] \leq \delta/(1-\exp(\eps - \eps'))$, and (b) $P \mid B^c \equiv_{\eps', 0} Q \mid B^c$. In particular, setting $\eps' = \eps + \sqrt{2delta}$, we get $\Pr_P[B] \leq \sqrt{\delta}$.
\end{lemma}
\begin{proof}
  Without loss of generality\footnote{This can be ensured by having the mechanism outputting a uniform $[0,1]$ r.v. in addition to its original output.}, the distributions have a density function. Let $B = \{x : \frac{\Pr_{P}[x]}{\Pr_Q[x]} \geq \exp(\eps') \}$. Now note that
  \begin{align*}
    \Pr_P[B] \leq \exp(\eps)\Pr_Q[B] + \delta
    &\leq \exp(\eps)\exp(-\eps')\Pr_P[B] + \delta,
  \end{align*}
  so that $\Pr_P[B] \leq \delta / (1-\exp(\eps - \eps'))$. Setting $\eps' = \eps + \sqrt{2\delta}$, and noting that $\exp(-\sqrt{2\delta}) \leq 1 - \sqrt{\delta}$ for $\delta < \frac 1 {10}$, the claim follows.

\end{proof}

\begin{lemma}
  \label{lem:cond_e_to_ed}
  Suppose that there is an event $B$ such that $P \mid B^c \equiv_{\eps, 0} Q \mid B^c$, and that $\Pr_{x \sim P}[B] \leq \delta$. Then $P \equiv_{\eps, \delta} Q$.
\end{lemma}
\begin{proof}
  Let $S$ be any event. Then
  \begin{align*}
    \Pr_{P}[S] &\leq \Pr_{P \mid B^c}[S] + \Pr[B]\\
    &\leq \exp(\eps) \Pr_{Q \mid B^c}[S] + \Pr[B]\\
    &\leq \exp(\eps) \Pr_{Q}[S] + \delta.\\
  \end{align*}
\end{proof}

%% file: main-split.bbl
\begin{thebibliography}{10}

\bibitem{abadi2016deep}
M.~Abadi, A.~Chu, I.~Goodfellow, H.~B. McMahan, I.~Mironov, K.~Talwar, and
  L.~Zhang.
\newblock Deep learning with differential privacy.
\newblock In {\em Proceedings of the 2016 ACM SIGSAC Conference on Computer and
  Communications Security}, pages 308--318. ACM, 2016.

\bibitem{BassilyNSSSU16}
R.~Bassily, K.~Nissim, A.~Smith, T.~Steinke, U.~Stemmer, and J.~Ullman.
\newblock Algorithmic stability for adaptive data analysis.
\newblock In {\em Proceedings of the forty-eighth annual ACM symposium on
  Theory of Computing}, pages 1046--1059. ACM, 2016.

\bibitem{BeimelNS13}
A.~Beimel, K.~Nissim, and U.~Stemmer.
\newblock Private learning and sanitization: Pure vs. approximate differential
  privacy.
\newblock In {\em Approximation, Randomization, and Combinatorial Optimization.
  Algorithms and Techniques}, pages 363--378, Berlin, Heidelberg, 2013.
  Springer Berlin Heidelberg.

\bibitem{BhaskarLST10}
R.~Bhaskar, S.~Laxman, A.~Smith, and A.~Thakurta.
\newblock Discovering frequent patterns in sensitive data.
\newblock In {\em Proceedings of the 16th ACM SIGKDD international conference
  on Knowledge discovery and data mining}, pages 503--512. ACM, 2010.

\bibitem{BunNSV15}
M.~Bun, K.~Nissim, U.~Stemmer, and S.~Vadhan.
\newblock Differentially private release and learning of threshold functions.
\newblock In {\em Proceedings of the 2015 IEEE 56th Annual Symposium on
  Foundations of Computer Science (FOCS)}, FOCS '15, pages 634--649,
  Washington, DC, USA, 2015. IEEE Computer Society.

\bibitem{negative_moments}
M.~T. Chao and W.~E. Strawderman.
\newblock Negative moments of positive random variables.
\newblock {\em Journal of the American Statistical Association},
  67(338):429--431, 1972.

\bibitem{ChaudhuriHS14}
K.~Chaudhuri, D.~Hsu, and S.~Song.
\newblock The large margin mechanism for differentially private maximization.
\newblock In {\em Proceedings of the 27th International Conference on Neural
  Information Processing Systems - Volume 1}, NIPS'14, pages 1287--1295,
  Cambridge, MA, USA, 2014. MIT Press.

\bibitem{ChaudhuriSS12}
K.~Chaudhuri, A.~Sarwate, and K.~Sinha.
\newblock Near-optimal differentially private principal components.
\newblock In {\em Advances in Neural Information Processing Systems}, pages
  989--997, 2012.

\bibitem{ChaudhuriV13}
K.~Chaudhuri and S.~A. Vinterbo.
\newblock A stability-based validation procedure for differentially private
  machine learning.
\newblock In {\em Advances in Neural Information Processing Systems}, pages
  2652--2660, 2013.

\bibitem{DworkFHPRR15}
C.~Dwork, V.~Feldman, M.~Hardt, T.~Pitassi, O.~Reingold, and A.~L. Roth.
\newblock Preserving statistical validity in adaptive data analysis.
\newblock In {\em Proceedings of the forty-seventh annual ACM symposium on
  Theory of computing}, pages 117--126. ACM, 2015.

\bibitem{DworkL09}
C.~Dwork and J.~Lei.
\newblock Differential privacy and robust statistics.
\newblock In {\em Proceedings of the forty-first annual ACM symposium on Theory
  of computing}, pages 371--380. ACM, 2009.

\bibitem{DMNS}
C.~Dwork, F.~McSherry, K.~Nissim, and A.~Smith.
\newblock Calibrating noise to sensitivity in private data analysis.
\newblock In {\em Proceedings of the Third Conference on Theory of Cryptography
  (TCC)}, pages 265--284, 2006.

\bibitem{DworkNRRV09}
C.~Dwork, M.~Naor, O.~Reingold, G.~N. Rothblum, and S.~Vadhan.
\newblock On the complexity of differentially private data release: efficient
  algorithms and hardness results.
\newblock In {\em Proceedings of the forty-first annual ACM symposium on Theory
  of computing}, pages 381--390. ACM, 2009.

\bibitem{dwork2014algorithmic}
C.~Dwork and A.~Roth.
\newblock The algorithmic foundations of differential privacy.
\newblock {\em Foundations and Trends in Theoretical Computer Science},
  9(3--4):211--407, 2014.

\bibitem{DworkSZ15}
C.~Dwork, W.~Su, and L.~Zhang.
\newblock Private false discovery rate control.
\newblock {\em arXiv preprint arXiv:1511.03803}, 2015.

\bibitem{GargK18}
V.~Garg and A.~Kalai.
\newblock Supervising unsupervised learning.
\newblock In {\em 32nd Conference on Neural Information Processing Systems
  (NIPS)}, 2018.
\newblock To Appear.

\bibitem{gelman2014statistical}
A.~Gelman and E.~Loken.
\newblock The statistical crisis in science.
\newblock {\em American scientist}, 102(6):460, 2014.

\bibitem{trunclap}
Q.~{Geng}, W.~{Ding}, R.~{Guo}, and S.~{Kumar}.
\newblock {Truncated Laplacian Mechanism for Approximate Differential Privacy}.
\newblock {\em ArXiv e-prints}, Oct. 2018.

\bibitem{GLMRT}
A.~Gupta, K.~Ligett, F.~McSherry, A.~Roth, and K.~Talwar.
\newblock Differentially private combinatorial optimization.
\newblock In {\em Proceedings of the Twenty-first Annual ACM-SIAM Symposium on
  Discrete Algorithms}, SODA '10, pages 1106--1125, Philadelphia, PA, USA,
  2010. Society for Industrial and Applied Mathematics.

\bibitem{KapralovT13}
M.~Kapralov and K.~Talwar.
\newblock On differentially private low rank approximation.
\newblock In {\em Proceedings of the twenty-fourth annual ACM-SIAM symposium on
  Discrete algorithms}, pages 1395--1414. SIAM, 2013.

\bibitem{KLNRS}
S.~P. Kasiviswanathan, H.~K. Lee, K.~Nissim, S.~Raskhodnikova, and A.~D. Smith.
\newblock What can we learn privately?
\newblock In {\em 49th Annual {IEEE} Symposium on Foundations of Computer
  Science, {FOCS}}, pages 531--540, 2008.

\bibitem{pythia}
I.~Kotsogiannis, A.~Machanavajjhala, M.~Hay, and G.~Miklau.
\newblock Pythia: Data dependent differentially private algorithm selection.
\newblock In {\em Proceedings of the 2017 ACM International Conference on
  Management of Data}, pages 1323--1337. ACM, 2017.

\bibitem{hyperband}
L.~Li, K.~Jamieson, G.~DeSalvo, A.~Rostamizadeh, and A.~Talwalkar.
\newblock Hyperband: Bandit-based configuration evaluation for hyperparameter
  optimization.
\newblock In {\em ICLR}, 2017.

\bibitem{LigettNRWW17}
K.~Ligett, S.~Neel, A.~Roth, B.~Waggoner, and S.~Z. Wu.
\newblock Accuracy first: Selecting a differential privacy level for accuracy
  constrained erm.
\newblock In {\em Advances in Neural Information Processing Systems}, pages
  2566--2576, 2017.

\bibitem{mcgregor2010limits}
A.~McGregor, I.~Mironov, T.~Pitassi, O.~Reingold, K.~Talwar, and S.~Vadhan.
\newblock The limits of two-party differential privacy.
\newblock In {\em Foundations of Computer Science (FOCS), 2010 51st Annual IEEE
  Symposium on}, pages 81--90. IEEE, 2010.

\bibitem{McSherryT}
F.~McSherry and K.~Talwar.
\newblock Mechanism design via differential privacy.
\newblock In {\em Annual IEEE Symposium on Foundations of Computer Science
  (FOCS)}. IEEE, October 2007.

\bibitem{MinamiASN16}
K.~Minami, H.~Arai, I.~Sato, and H.~Nakagawa.
\newblock Differential privacy without sensitivity.
\newblock In {\em Proceedings of the 30th International Conference on Neural
  Information Processing Systems}, NIPS'16, pages 964--972, USA, 2016. Curran
  Associates Inc.

\bibitem{mir2013}
D.~J. Mir.
\newblock {\em Differential privacy: an exploration of the privacy-utility
  landscape}.
\newblock PhD thesis, Rutgers University-Graduate School-New Brunswick, 2013.

\bibitem{NissimRS07}
K.~Nissim, S.~Raskhodnikova, and A.~Smith.
\newblock Smooth sensitivity and sampling in private data analysis.
\newblock In {\em Proceedings of the thirty-ninth annual ACM symposium on
  Theory of computing}, pages 75--84. ACM, 2007.

\bibitem{papernot2016semi}
N.~Papernot, M.~Abadi, U.~Erlingsson, I.~Goodfellow, and K.~Talwar.
\newblock Semi-supervised knowledge transfer for deep learning from private
  training data.
\newblock {\em arXiv preprint arXiv:1610.05755}, 2016.

\bibitem{RaskhodnikovaS16}
S.~Raskhodnikova and A.~Smith.
\newblock Lipschitz extensions for node-private graph statistics and the
  generalized exponential mechanism.
\newblock In {\em Foundations of Computer Science (FOCS), 2016 IEEE 57th Annual
  Symposium on}, pages 495--504. IEEE, 2016.

\bibitem{Smith11}
A.~Smith.
\newblock Privacy-preserving statistical estimation with optimal convergence
  rates.
\newblock In {\em Proceedings of the Forty-third Annual ACM Symposium on Theory
  of Computing}, STOC '11, pages 813--822, New York, NY, USA, 2011. ACM.

\bibitem{SmithT13}
A.~Smith and A.~Thakurta.
\newblock Differentially private feature selection via stability arguments, and
  the robustness of the lasso.
\newblock In S.~Shalev-Shwartz and I.~Steinwart, editors, {\em Proceedings of
  the 26th Annual Conference on Learning Theory}, volume~30 of {\em Proceedings
  of Machine Learning Research}, pages 819--850, Princeton, NJ, USA, 12--14 Jun
  2013. PMLR.

\bibitem{UhleropSF13}
C.~Uhlerop, A.~Slavkovi{\'c}, and S.~E. Fienberg.
\newblock Privacy-preserving data sharing for genome-wide association studies.
\newblock {\em The Journal of privacy and confidentiality}, 5(1):137, 2013.

\bibitem{wiki:hyperparameter}
Wikipedia.
\newblock {Hyperparameter optimization} --- {W}ikipedia{,} the free
  encyclopedia.
\newblock
  \url{http://en.wikipedia.org/w/index.php?title=Hyperparameter\%20optimization&oldid=866173926},
  2018.
\newblock [Online; accessed 30-October-2018].

\end{thebibliography}
